\documentclass[11pt,reqno,a4paper]{amsart}
\pdfoutput=1
\usepackage{bold-extra}
\usepackage{microtype}
\usepackage{amsfonts,amstext,amsmath,amssymb,stmaryrd,calc,mathrsfs,booktabs}
\def\mathrlapinternal#1#2{\rlap{$\mathsurround=0pt#1{#2}$}}
\def\mathrlap{\mathpalette\mathrlapinternal}
\newcommand{\rua}[1]{\mathrel{\mathrlap{\:\xrightarrow{{\color{white} #1}}}{\xrightarrow{#1}\:}}}
\newcommand{\eqdef}{\stackrel{\raisebox{-.15ex}{\scalebox{.5}{\upshape\textrm{def}}}}{=}}
\def\vec#1{\mathchoice{\mbox{\boldmath$\displaystyle#1$}}
{\mbox{\boldmath$\textstyle#1$}}
{\mbox{\boldmath$\scriptstyle#1$}}
{\mbox{\boldmath$\scriptscriptstyle#1$}}}
\newcommand{\norm}[1]{\|#1\|}
\newcommand{\tup}[1]{\langle #1\rangle}

\newcommand{\dc}{\mathop{\downarrow}\!}
\newcommand{\uc}{\mathop{\uparrow}\!}
\newcommand{\ra}{\mathrel{\rightarrow^\ast}}
\newcommand{\ru}[1]{\xrightarrow{#1}}
\newcommand{\dom}{\mathop{\mathsf{dom}}}
\newcommand{\remove}[1]{\overline{#1}}
\newcommand{\prefix}[1]{\mathsf{Pref}(#1)}
\newcommand{\nat}{\mathbb{N}}
\newcommand{\infwords}{\Sigma^\ast_\mathsf{acc}}

\newcommand{\Ltrans}{L_{\mathsf{acc}}}

\newcommand{\subword}{\preceq}
\newcommand{\tickYes}{Yes}%
\newcommand{\tickNo}{No}%
\usepackage{color,graphicx,tikz}
\usetikzlibrary{petri}
\usetikzlibrary{positioning}
\usetikzlibrary{arrows}
\usetikzlibrary{automata}
\usetikzlibrary{trees}
\usepackage{multirow,longtable}
\usepackage{listings}
\newcommand{\cpp}{\lstinline[language=C++]}
\usepackage[numbers]{natbib}\renewcommand{\cite}{\citep}

\makeatletter
\def\NAT@spacechar{~}%
\makeatother
\renewcommand{\cite}{\citep}

\providecommand{\doi}[1]{doi:\href{http://dx.doi.org/#1}{\nolinkurl{#1}}}
\makeatletter
\newcommand{\citepay}[2][\@empty]{\citeauthor{#2}~(\ifx#1\@empty\relax\else#1
  \fi\citeyear{#2})}
\def\bibinfo@X@doi#1{#1}
\providecommand{\doi}[1]{doi:\href{http://dx.doi.org/#1}{\nolinkurl{#1}}}
\makeatother
\newcommand{\citeay}[1]{\citeauthor{#1}, \citeyear{#1}}
\usepackage{amsthm,thmtools}
\theoremstyle{plain}
\newtheorem{theorem}{Theorem}
\newtheorem{lemma}[theorem]{Lemma}
\newtheorem{proposition}[theorem]{Proposition}
\newtheorem{corollary}[theorem]{Corollary}
\newtheorem{fact}[theorem]{Fact}
\theoremstyle{definition}
\newtheorem{definition}[theorem]{Definition}

\theoremstyle{remark}
\newtheorem{claim}{Claim}[theorem]
\providecommand{\qedhere}{\qed}
\renewcommand{\paragraph}{\subsubsection*}
\usepackage{url}
\usepackage{hyperref}

\renewcommand{\subsectionautorefname}{Section}
\renewcommand{\subsubsectionautorefname}{Section}
\renewcommand{\sectionautorefname}{Section}
\renewcommand{\subsectionautorefname}{Section}
\renewcommand{\subsubsectionautorefname}{Section}

\providecommand{\subfigureautorefname}{Figure}
\begin{document}
\providecommand{\subfigureautorefname}{Fig.$\!$}
\def\sectionautorefname{Section} \def\subsectionautorefname{Section}
\renewcommand\subsubsectionautorefname[1]{\S}
\def\subfigureautorefname{Figure} \def\chapterautorefname{Chapter}

\title{Forward Analysis and Model Checking \mbox{for Trace Bounded} WSTS}
\thanks{Work supported by ANR projects AVeriSS (ANR-06-SETIN-001) and
  AVERILES (ANR-05-RNTL-02).  An extended abstract of this work
  appeared in the Proceedings of the 32nd International Conference on
  Application and Theory of Petri Nets, \emph{Lect. Notes in
    Comput. Sci.} vol.~6709 pp.~49--68 (L.\ M.~Kristensen and
  L.~Petrucci, Eds.), Springer, 2011.}  \author[P.~Chambart]{Pierre
  Chambart} \address{OCamlPro} \author[A.~Finkel]{Alain Finkel}
\author[S.~Schmitz]{Sylvain Schmitz} \address{LSV, ENS Cachan \& CNRS
  \& INRIA, Universit\'e Paris-Saclay, France}
\begin{abstract}
  We investigate a subclass of well-structured transition systems
  (WSTS), the \emph{trace bounded}---in the sense of \citeauthor{bcfl}
  (\emph{T.~Amer. Math. Soc.}, 1964)---complete deterministic ones, which we
  claim provide an adequate basis for the study of forward analyses as
  developed by \citeauthor{cwsts2} (\emph{Logic. Meth. Comput. Sci.},
  2012).  Indeed, we prove that, unlike other conditions considered
  previously for the termination of forward analysis, trace
  boundedness is decidable.  Trace boundedness turns out to be a
  valuable restriction for WSTS verification, as we show that it
  further allows to decide all $\omega$-regular properties on the set
  of infinite traces of the system.  \keywords complete WSTS, model
  checking, flattable system, bounded language, acceleration
\end{abstract}
\maketitle

\section{Introduction}
\paragraph{General Context} Forward analysis using acceleration
\citep{accel,flataccel} is established as one of the most efficient
practical means---albeit in general without termination guarantee---to
tackle safety problems in infinite state systems, e.g.\ in the tools
\textsc{TReX}~\citep{trex}, \textsc{Lash}~\citep{lash}, or
\textsc{Fast}~\citep{fast}.  Even in the context of
\emph{well-structured transition systems} (WSTS), a unifying framework
for infinite systems endowed with a generic backward coverability
algorithm due to \citet{wqo}, forward procedures are commonly felt to
be more efficient than the backward algorithm~\citep{prevspost}: e.g.\
for lossy channel systems~\citep{lcs}, although the backward procedure
always terminates, only the non terminating forward procedure is
implemented in the tool \textsc{TReX}~\citep{trex}.

Acceleration techniques rely on symbolic representations of sets of
states to compute exactly the effect of repeatedly applying a finite
sequence of transitions $w$, i.e.\ the effect of $w^\ast$.  The forward
analysis terminates if and only if a finite sequence $w_1^\ast\cdots
w_n^\ast$ of such accelerations deriving the full reachability set can
be found, resulting in the definition of the \emph{post$^\ast$
  flattable} class of systems \citep{flataccel}.  Despite evidence
that many classes of systems are
flattable~\citep{flat2dim,flatevery,BFGHM-lics15},
whether a system is post$^\ast$-flattable is undecidable for general
systems~\citep{flataccel}.
 
\paragraph{The Well Structured Case}
\Citet{cwsts1,cwsts2} have laid new theoretical foundations for the
forward analysis of deterministic WSTS---where determinism is
understood with respect to transition labels---, by defining
\emph{complete} deterministic WSTS (cd-WSTS) as a means to obtain
finite representations for downward closed sets of states \citep[see
also][]{wadl}, \emph{$\infty$-effective} cd-WSTS as those for which
the acceleration of certain sequences can effectively be computed, and
by proposing a conceptual forward procedure \`a la
\citet{kmtree} for computing the full cover of a
cd-WSTS---i.e.\ the downward closure of its set of reachable states.
Similarly to post$^\ast$ flattable systems, this procedure called
``$\mathsf{Clover}$'' terminates if and only if the cd-WSTS at hand is
\emph{cover flattable}, which is undecidable~\citep{cwsts2}.  As we
show in this paper, post$^\ast$ flattability is also undecidable for
cd-WSTS, thus motivating the search for even stronger sufficient
conditions for termination.  A decidable sufficient condition that we
can easily discard as too restrictive is trace set finiteness,
corresponding to terminating systems~\citep{origwsts}.

\paragraph{This Work}
Our aim with this paper was to find a reasonable decidable sufficient
condition for the termination of the $\mathsf{Clover}$ procedure.  We
have found one such condition in the work of \citet*{foctlpr} with
\emph{trace flattable} systems, which are maybe better defined as the
systems with a \emph{bounded} trace set in the sense of~\citet{bcfl}:
a language $L\subseteq\Sigma^\ast$ is bounded if there exists
$n\in\mathbb{N}$ and $n$ words $w_1,\dots,w_n$ in $\Sigma^\ast$ such
that $L\subseteq w_1^\ast\cdots w_n^\ast$.  The regular expression
$w_1^\ast\cdots w_n^\ast$ is called a \emph{bounded expression} for
$L$.  Trace bounded cd-WSTS encompass systems with finite trace set.

Trace boundedness implies post$^\ast$ and cover flattability.
Moreover, \citeauthor{foctlpr} show that it allows to decide liveness
properties for a restricted class of counter systems (see also
\cite{demri15,DBLP:conf/fct/GantyI15} for other classes of trace
bounded counter systems).  However, to the best of our knowledge,
nothing was known regarding the decidability of trace boundedness
itself, apart from the \citeyear{bcfl} proof of decidability for
context-free grammars by \citet{bcfl} and the \citeyear{emg} one for
equal matrix grammars by \citet{emg}.

We characterize trace boundedness for cd-WSTS and provide as our
main contribution a generic decision algorithm in \autoref{sec:dec}.
We employ vastly different techniques than those used by \citet{bcfl}
and \citet{emg}, since we rely on the results of \citet{cwsts1,cwsts2}
to represent the effect of certain transfinite sequences of
transitions.  We further argue in \autoref{sec:undec} that both the
class of systems (deterministic WSTS) and the property (trace boundedness)
are in some sense optimal: we prove that trace boundedness becomes
undecidable if we relax either the determinism or the
well-structuredness conditions, and that the less restrictive property
of post$^\ast$ flattability is not decidable on deterministic WSTS.

We investigate in \autoref{sec:cmplx} the complexity of trace
boundedness.  It can grow very high depending on the type of
underlying system, but this is the usual state of things with
WSTS---e.g.\ the non multiply-recursive lower bound for coverability
in lossy channel systems of \citet{CS-lics08} also applies to trace
boundedness---and does not prevent tools to be efficient on case
studies.  Although there is no hope of finding general upper bounds
for all WSTS, we nevertheless propose a generic proof recipe, based on
a detailed analysis of our decidability proof, which results in tight
upper bounds in the cases of lossy channel systems and affine counter
systems.  In the simpler case of Petri nets, we demonstrate that trace
boundedness is \textsc{ExpSpace}-hard (matching the \textsc{ExpSpace}
upper bound from~\citep{Schmitz11}), but that the size of the
associated bounded expression can be non primitive-recursive.

Beyond coverability, and as further evidence to the interest of trace
boundedness for the verification of WSTS, we show that all
$\omega$-regular word properties can be checked against the set of
infinite traces of trace bounded $\infty$-effective cd-WSTS, resulting
in a non trivial recursive class of WSTS with decidable liveness
(\autoref{sub:declive}).  Liveness properties are in general
undecidable in cd-WSTS~\citep{undecLCS,lrpn}: techniques for
propositional linear-time temporal logic (LTL) model checking are not
guaranteed to terminate~\citep{emerson,fwlcs} or limited to
subclasses, like Petri nets~\citep{ltlpn}.  %
As a consequence of our result, action-based
(aka transition-based) LTL model checking is decidable for cd-WSTS
(\autoref{sub:decltl}), whereas state-based properties are
undecidable for trace bounded cd-WSTS \citep{racs}.
 
One might fear that trace boundedness is too strong a property to be of any
practical use.  For instance, commutations, as created by concurrent
transitions, often result in trace unboundedness.  However, bear in
mind that the same issues more broadly affect all forward analysis
techniques, and have been alleviated in tools through various
heuristics.  Trace boundedness offers a new insight into why such
heuristics work, and can be used as a theoretical foundation for their
principled development; we illustrate this point in \autoref{sec:tunb}
where we introduce trace boundedness modulo a partial commutation
relation.  We demonstrate the interest of this extension by verifying
a liveness property on the Alternating Bit Protocol with a bounded
number of sessions.

This work results in an array of concrete classes of WSTS, including
lossy channel systems \cite{lcs}, broadcast protocols \cite{emerson},
and Petri nets and their monotone extensions, such as reset/transfer
Petri nets \cite{boundedRPN}, for which trace boundedness is decidable
and implies both computability of the full coverability set and
decidability of liveness properties.  Even for trace unbounded
systems, it provides a new foundation for the heuristics currently
employed by tools to help termination, as with the commutation
reductions we just mentioned.

\renewcommand{\paragraph}{\subsubsection}
\section{Background}\label{sec:prelim}

\subsection{A Running Example}
\newcommand{\semicommutation}[1]{}
\begin{figure}[t]%
\centering
\begin{minipage}{.6\textwidth}%
\footnotesize
\begin{lstlisting}
// Performs n invocations of the rpc() function
// with at most P>=1 simultaneous concurrent calls
piped_multirpc (int n) {
  int sent = n, recv = n; rendezvous rdv;
  while (recv > 0)
    if (sent > 0 && recv - sent < P) {
      post(rdv, rpc); // asynchronous call
      sent--;
    } else { // sent == 0 || recv - sent >= P
      wait(rdv); // a rpc has returned
      recv--;
    }
}
\end{lstlisting}%
\end{minipage}%
\\
\begin{minipage}{\textwidth}
\centering\tikzstyle{token}=[fill=gray,draw=none,circle,inner sep=0.5pt,minimum size=1ex,text=white,font=\pgfutil@font@tiny,every token]
\begin{tikzpicture}[->,>=stealth',shorten >=1pt,initial text=,%
                    node distance=3cm,on grid,semithick,auto,
                    inner sep=2pt,every transition/.style={minimum
                      width=2mm,minimum height=8mm}]
  \node[place,label={[gray]above:\cpp{main}},color=gray,label={[gray,font=\tiny]left:$(1)$}](0){};
  \node[fill=gray,circle,inner sep=0.5pt,minimum size=1ex]{};
  \node[place,below right=2.8cm and 2cm of 0,label={below:\cpp{n}},label={[font=\tiny]above:$(3)$}](2){};
  \node[transition,below=1.5cm of 0,minimum width=8mm,minimum height=2mm,label={[gray,font=\footnotesize]left:$g$},color=gray]{}
    edge[pre,bend left=20,color=gray](0) edge[post,bend right=20,color=gray](0) edge[post,bend right=35,color=gray](2);
  \node[place,right=2cm of 0,label
    distance=-3ex,label={above:\cpp{piped_multirpc}},label={[font=\tiny]below left:$(2)$}](1){};
  \node[place,below right=.8cm and 3cm of 1,label={above:\cpp{P-recv+sent}},label={[font=\tiny]right:$(4)$}](3){\cpp{P}};
  \node[place,above right=.8cm and 3cm of 2,label={below:\cpp{recv-sent}},label={[font=\tiny]right:$(5)$}](4){};
  \node[transition,right=1cm of 0,label={[gray,font=\footnotesize]below:$c$},color=gray]{}
    edge[pre,color=gray](0) edge[post,color=gray](1);
  \semicommutation{%
  \node[place,color=gray,dashed,above right=.5cm and 4.2cm of
    1,label={[gray]above left:\cpp{rdv}},label={[gray,font=\tiny]above right:$(7)$}](rdv){};
  \node[place,color=gray,dashed,below right=.5cm and 4.2cm of
    2,label={[gray]below left:\cpp{rpc}},label={[gray,font=\tiny]below right:$(6)$}](rpc){};
  }
  \node[transition,right=1.5cm of 1,label={[font=\footnotesize]above:$e$}]{}
    edge[pre,bend right=20](1) edge[post,bend left=20](1)
    edge[pre,bend right=20](4) edge[post,bend right=10](3)
    \semicommutation{edge[pre,bend left=20,color=gray,dashed](rdv)};
  \node[transition,right=1.5cm of 2,label={[font=\footnotesize]below:$i$}]{}
    edge[pre,bend right=20](1) edge[post, bend left=20](1) edge[pre](2)
    edge[pre,bend left=20](3) edge[post,bend left=10](4)
    \semicommutation{edge[post,bend right=20,color=gray,dashed](rpc)};
  \semicommutation{%
    \node[transition,color=gray,dashed,below=1.9cm of rdv,minimum width=8mm,minimum
      height=2mm,label={[font=\footnotesize,gray]right:$r$}]{}
    edge[pre,color=gray,dashed](rpc)
    edge[post,color=gray,dashed](rdv);
  }
  \end{tikzpicture}
\end{minipage}\vspace{-1em}
\caption{\label{fig:petri}A piped RPC client in C-like syntax, and its Petri net modelization.}
\end{figure}
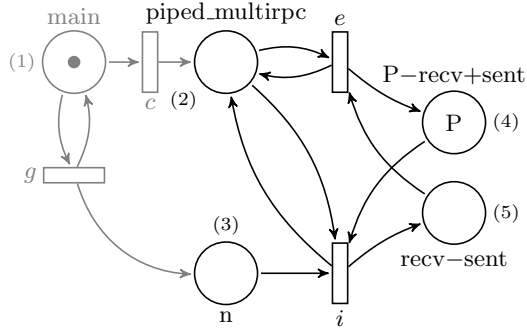

\lstset{basicstyle=\small}
We consider throughout this paper an example (see \autoref{fig:petri})
inspired by the recent literature on \emph{asynchronous} or
\emph{event-based} programming \citep{tame,liveasynch}, namely that of
a client performing $n$ asynchronous remote procedure calls
(corresponding to the \cpp{post(r,rpc)} statement on line~7), of which
at most $P$ can simultaneously be pending.  Such piped---or
windowed---clients are commonly employed to prevent server
saturation.

The abstracted ``producer/consumer'' Petri net for this program
(ignoring the grayed parts for now) has two transitions $i$ and $e$
modeling the \cpp{if} and \cpp{else} branches of lines 6 and 9
respectively.  The deterministic choice between these two branches is
here replaced by a nondeterministic one, where the program can choose
the \cpp{else} branch and wait for some \cpp{rpc} call to return
before the window of pending calls is exhausted.  Observe that we can
recover the original program behavior by further controlling the Petri
net with the \emph{bounded} regular language $i^P(ei)^\ast e^P$ ($P$
is fixed), i.e.\ taking the intersection by synchronous product with a
deterministic finite automaton for $i^P(ei)^\ast e^P$.  
This is an example of a trace bounded system. 

Even without bounded control, the Petri net of
\autoref{fig:petri} has a bounded, finite, language for each
fixed initial $n$; however, for $P\geq 2$, if we expand it for parametric
verification with the left grayed area to allow any $n$ (or set $n=\omega$
as initial value to switch to server mode), then its language becomes
unbounded.  We will reuse this example in \autoref{sec:dec} when
characterizing unboundedness in cd-WSTS.  The full system is of course
bounded when synchronized with a deterministic finite automaton for
the language $g^\ast ci^P(ei)^\ast e^P$.

\subsection{Definitions}

\paragraph{Languages}
Let $\Sigma$ be a finite alphabet; we denote by $\Sigma^\ast$ the set
of finite sequences of elements from $\Sigma$, and by $\Sigma^\omega$
that of infinite sequences;
$\Sigma^\infty\eqdef\Sigma^\ast\cup\Sigma^\omega$.  We denote the
empty sequence by $\varepsilon$, the set of non empty finite sequences
by $\Sigma^+\eqdef\Sigma^\ast\setminus\{\varepsilon\}$, the length of
a sequence $w$ by $|w|$, the left quotients of a language
$L_2\subseteq L^\infty$ by a language $L_1\subseteq\Sigma^\ast$ by
$L_1^{-1}L_2\eqdef\{v\in\Sigma^\infty\mid\exists u\in L_1,uv\in
L_2\}$, and the set of finite prefixes of $L_2$ by
$\prefix{L_2}\eqdef\{u\in\Sigma^\ast\mid\exists
v\in\Sigma^\infty,uv\in L\}$.

We make regular use of the closure of bounded languages by finite
union, intersection and concatenation, taking subsets, prefixes,
suffixes, and factors, and of the following sufficient condition for the
unboundedness of a language $L$ \citep[Lemma~5.3]{bcfl}: the existence
of two words $u$ and $v$ in $\Sigma^+$, such that $uv\neq vu$ and each
word in $\{u,v\}^\ast$ is a factor of some word in~$L$.

\paragraph{Orderings}
Given a relation $\mathrel{R}$ on $A\times B$, we denote by
$\mathrel{R}^{-1}$ its inverse, by $\mathbin{R}(C)\subseteq B$ the
image of $C\subseteq A$, by $\mathbin{R}^\ast$ its transitive reflexive
closure if $\mathbin{R}(A)\subseteq A$, and by
$\mathop{\mathsf{dom}}\mathrel{R}\eqdef{\mathbin{R}}^{-1}(B)$ its
domain.

A \emph{quasi ordering} $\leq$ is a reflexive and transitive relation on
a set $S$.  We write $\geq\;=\;\leq^{-1}$ for the converse quasi order,
${<}\eqdef{\leq}\setminus{\geq}$ for the associated strict order, and
${\equiv}\eqdef {\leq}\cap{\leq}^{-1}$ for the associated equivalence relation.  
The
\emph{$\leq$-upward closure} $\uc C$ of a set $C\subseteq S$ is
$\{s\in S\mid\exists c\in C,c\leq s\}$; its \emph{$\leq$-downward
closure} is $\dc C\eqdef\{s\in S\mid\exists c\in C,c\geq s\}$.  A set
$C$ is \emph{$\leq$-upward closed} (resp.\ $\leq$-downward closed) if
$\uc C=C$ (resp.\ $\dc C=C$).  A set $B$ is a \emph{basis} for an
upward-closed set $C$ (resp.\ downward-closed) if $\uc B=C$
(resp.\ $\dc B=C$).  An upper bound $s\in S$ of a set $A$ verifies
$a\leq s$ for all $a$ of $A$, while we denote its \emph{least upper
  bound}, if it exists, by $\mathsf{lub}(A)$.

A \emph{well quasi ordering} (wqo) is a quasi ordering such that for
any infinite sequence $s_0s_1s_2\cdots$ of $S^\omega$ there exist $i<j$ in
$\mathbb{N}$ such that $s_i\leq s_j$.  Equivalently, there does not
exist any strictly descending chain $s_0>s_1>\cdots>s_i>\cdots$, and
any \emph{antichain}, i.e.\ set of pairwise incomparable elements, is
finite.  In particular, the set of minimal elements of an
upward-closed set $C$ is finite when quotiented by $\equiv$, and is a basis for $C$.
Pointwise comparison $\leq$ in $\mathbb{N}^k$, and scattered subword
comparison $\preceq$ on finite sequences in $\Sigma^\ast$ are well
quasi orders by Higman's Lemma.

\paragraph{Continuous Directed Complete Partial Orders}
A \emph{directed subset} $D\neq\emptyset$ of $S$ is such that any pair
$\{x,y\}$ of elements of $D$ has an upper bound in $D$.  A
\emph{directed complete partial order} (dcpo) is such that any
directed subset has a least upper bound.  A subset $O$ of a dcpo is
\emph{open} if it is upward-closed and if, for any directed subset $D$
such that $\mathsf{lub}(D)$ is in $O$, $D\cap O\neq\emptyset$.  A
partial function $f$ on a dcpo is \emph{partial continuous} if it is
monotonic, $\mathsf{dom}f$ is open, and for any directed subset $D$ of
$\mathsf{dom}f$, $\mathsf{lub}(f(D))=f(\mathsf{lub}(D))$.  Two
elements $s$ and $s'$ of a dcpo are in a \emph{way below} relation,
noted $s\ll s'$, if for every directed subset $D$ such that
$\mathsf{lub}(D)\leq s'$, there exists $s''\in D$ s.t.\ $s\leq s''$.
A dcpo is \emph{continuous} if, for every $s'$ in $S$,
$\mathsf{wb}(s')\eqdef\{s\in S\mid s\ll s'\}$ is directed and has $s'$
as least upper bound.

\paragraph{Well Structured Transition Systems}
A \emph{labeled transition system} (LTS)
$\mathcal{S}=\tup{S,s_0,\Sigma,\rightarrow}$ comprises a set $S$ of
states, an initial state $s_0\in S$, a finite set of labels $\Sigma$, a
transition relation $\rightarrow$ on $S$ defined as the union of the
relations ${\ru{a}}\subseteq S\times S$ for each $a$ in $\Sigma$.  The
relations are extended to sequences in $\Sigma^\ast$ by
$s\ru{\varepsilon}s$ and $s\ru{aw}s''$ for $a$ in $\Sigma$ and $w$ in
$\Sigma^\ast$ if there exists $s'$ in $S$ such
that $s\ru{a}s'$ and $s'\ru{w}s''$.  We write $\mathcal{S}(s)$ for the
same LTS with $s$ in $S$ as initial state (instead of $s_0$).  A LTS is
\begin{itemize}
\item \emph{uniformly bounded branching} if there exists $k \in \mathbb{N}$ such that
  $\mathsf{Post}_{\mathcal{S}}(s)\eqdef\{s'\in S\mid s\ru{}s'\}$ contains less than $k$ elements
  for all $s$ in~$S$,
\item \emph{deterministic} if $\ru{a}$ is a partial function for each
  $a$ in $\Sigma$---and is thus uniformly bounded branching---; we abuse
  notation in this case and identify $u$ with the partial function
  $\ru{u}$ for $u$ in $\Sigma^\ast$,
\item \emph{state bounded} if its \emph{reachability set}
  $\mathsf{Post}_{\mathcal{S}}^\ast(s_0)\eqdef\{s\in S\mid s_0\ra s\}$ is finite,
\item \emph{trace bounded} if its \emph{trace set}
  $T(\mathcal{S})\eqdef\{w\in\Sigma^\ast\mid\exists
    s\in S,s_0\ru{w}s\}$ is a bounded language,
\item \emph{terminating} if its trace set $T(\mathcal{S})$ is finite.
\end{itemize}

A \emph{well-structured transition system} (WSTS)~\citep{origwsts,wqo,wsts}
$\tup{S,s_0,\Sigma,{\rightarrow},{\leq},F}$ is a labeled
transition system $\tup{S,s_0,\Sigma,{\rightarrow}}$ 
endowed with a wqo
$\leq$ on $S$ and an $\leq$-upward closed set of final states $F$, such
that $\rightarrow$ is \emph{monotonic} wrt.\ $\leq$: for any $s_1$,
$s_2$, $s_3$ in $S$ and $a$ in $\Sigma$, if $s_1\leq s_2$ and
$s_1\ru{a}s_3$, then there exists $s_4\geq s_3$ in $S$ with
$s_2\ru{a}s_4$.

The \emph{language} of a WSTS is defined as
$L(\mathcal{S})\eqdef\{w\in\Sigma^\ast\mid\exists s\in F,s_0\ru{w}s\}$; see
\citet{wsl} for a general study of such languages.  In the context of
Petri nets, $L(\mathcal{S})$ is also called the \emph{covering} or \emph{weak}
language, and $T(\mathcal{S})$ the \emph{prefix} language.  Observe that a
\emph{deterministic finite-state automaton} (DFA) is a deterministic
WSTS $\mathcal{A}=\tup{Q,q_0,\Sigma,\delta,{=},F}$, where $Q$ is
finite (we shall later omit $=$ from the definition of DFAs).

Given
$\mathcal{S}_1=\tup{S_1,s_{0,1},\Sigma,{\rightarrow_1},{\leq_1},F_1}$
and
$\mathcal{S}_2=\tup{S_2,s_{0,2},\Sigma,{\rightarrow_2},{\leq_2},F_2}$ two WSTS,
their \emph{synchronous product} is the WSTS
$\mathcal{S}_1\times\mathcal{S}_2\eqdef\tup{S_1\times
  S_2,(s_{0,1},s_{0,2}),\Sigma,{\rightarrow_\times},{\leq_\times},F_1\times
  F_2}$, where for all $s_1$, $s'_1$ in $S_1$, $s_2$, $s'_2$ in $S_2$,
$a$ in $\Sigma$, $(s_1,s_2)\ru{a}_\times(s'_1,s'_2)$ if and only if
$s_1\ru{a}_1 s'_1$ and $s_2\ru{a}_2 s'_2$, and
$(s_1,s_2)\leq_\times(s'_1,s'_2)$ if and only if $s_1\leq_1s'_1$ and
$s_2\leq_2s'_2$, is again a WSTS, such that
$L(\mathcal{S}_1\times\mathcal{S}_2)=L(\mathcal{S}_1)\cap
L(\mathcal{S}_2)$.

We often consider the case $F=S$ and omit $F$ from the WSTS definition,
as we are more interested in trace sets, which provide more evidence
on the reachability sets.

\paragraph{Coverability}
A WSTS is \emph{$\mathsf{Pred}$-effective} if $\rightarrow$ and $\leq$
are decidable, and a finite basis for
$\uc\mathsf{Pred}_{\mathcal{S}}(\uc s, a)\eqdef\uc\{s'\in S\mid\exists
s''\in S,s'\ru{a}s''\text{ and }s\leq s''\}$ can effectively be
computed for all $s$ in $S$ and $a$ in $\Sigma$ \cite{wsts}.

The \emph{cover set} of a WSTS is
$\mathsf{Cover}_{\mathcal{S}}(s_0)\eqdef\dc\mathsf{Post}^\ast_{\mathcal{S}}(s_0)$,
and it is decidable whether a given state $s$ belongs to
$\mathsf{Cover}_{\mathcal{S}}(s_0)$ for finite branching
$\mathsf{Pred}$-effective WSTS, thanks to a backward algorithm that
checks whether $s_0$ belongs to
$\uc\mathsf{Pred}^\ast_{\mathcal{S}}(\uc s)\eqdef\uc\{s'\in
S\mid\exists s''\in S,s'\ra s''\text{ and }s''\geq s\}$.  One can also
decide the emptiness of the language of a WSTS, by checking whether
$s_0$ belongs to $\uc\mathsf{Pred}^\ast_{\mathcal{S}}(F)$.

\paragraph{Flattenings}
Let $\mathcal{A}$ be a DFA with a bounded language.  The synchronous
product $\mathcal{S}\times\mathcal{A}$ of $\mathcal{S}$ and
$\mathcal{A}$ is a \emph{flattening} of $\mathcal{S}$.  Consider the
projection $\pi$ from $S\times Q$ to $S$ defined by $\pi(s,q)\eqdef
s$; then $\mathcal{S}$ is \emph{post$^\ast$ flattable} if there exists
a flattening $\mathcal{S}'$ of $\mathcal{S}$ such that
$\mathsf{Post}^\ast_{\mathcal{S}}(s_0)=\pi(\mathsf{Post}^\ast_{\mathcal{S}'}((s_0,q_0)))$.
In the same way, it is \emph{cover flattable} if
$\mathsf{Cover}_{\mathcal{S}}(s_0)=\pi(\mathsf{Cover}_{\mathcal{S}'}((s_0,q_0)))$,
and \emph{trace flattable} if $T(\mathcal{S})=T(\mathcal{S}')$.
Remark that
\begin{enumerate}
\item trace flattability is equivalent to the boundedness of the trace
  set, and that
\item trace flattability implies post$^\ast$ flattability, which
in turn implies cover flattability.
\end{enumerate}

\paragraph{Complete WSTS}
A deterministic WSTS $\tup{S,s_0,\Sigma,{\rightarrow},{\leq}}$
is \emph{complete} (a cd-WSTS) if \mbox{$(S,\leq)$} is a continuous
dcpo and each transition function $a$ for $a$ in $\Sigma$ is partial continuous
\cite{cwsts1,cwsts2}.  The \emph{lub-acceleration} $u^\omega$ of
a partial continuous function $u$ on $S$, $u$ in $\Sigma^+$, is
again a partial function on $S$ defined by
\begin{align*}
\dom u^\omega&\eqdef\{s\in\dom u\:\mid s\leq u(s)\}\\
u^\omega(s)&\eqdef\mathsf{lub}(\{u^n(s)\mid
n\in\mathbb{N}\})&\text{for all $s$
  in $\dom u^\omega$.}
\end{align*}
A complete WSTS is \emph{$\infty$-effective} if
$u^\omega$ is computable for every $u$ in $\Sigma^+$.

\subsection{Working Hypotheses}\label{sub:hyp}
Our decidability results rely on some effectiveness assumptions for a
restricted class of WSTS: the complete deterministic ones.  We discuss
in this section the exact scope of these hypotheses.  As an appetizer,
notice that both trace boundedness and action-based $\omega$-regular
properties are only concerned with trace sets, hence one can more
generally consider classes of WSTS for which a trace-equivalent complete
deterministic system can effectively be found.  \autoref{fig:nom}
presents the various classes of systems mentioned at one point or
another in the main text or in the proofs.  It also provides a good
way to emphasize the applicability of our results on
$\infty$-effective cd-WSTS.

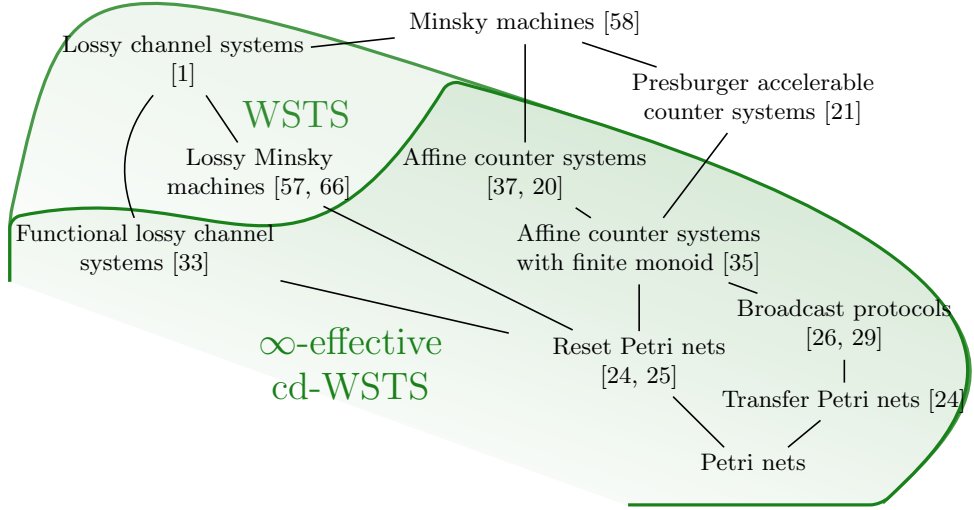
\begin{figure*}[t!]
  \centering
  \pgfdeclarelayer{background}\pgfsetlayers{background,main}
  \definecolor{mydarkgreen}{rgb}{0.1,.5,0.1}
  \begin{tikzpicture}[-,shorten >=1pt,initial text=,%
                     node distance=2cm,on grid,semithick,auto,
                     inner sep=2pt,
                     every node/.style={text width=3.2cm,text badly
                     centered,font=\footnotesize}]
    \node(minsky){Minsky machines {\citep{minsky}}};
    \node[below=2cm of minsky](aff){Affine counter
                     systems {\citep{affine,racs}}};
    \node[below right=1cm and 3cm of minsky](pres){Presburger
                     accelerable counter systems {\citep{foctlpr}}};
    \node[below left=0.5cm and 4.5cm of minsky](lcs){Lossy channel
                     systems {\citep{lcs}}};
    \node[below right=1.5cm and 1cm of lcs](lcn){Lossy Minsky
                     machines {\citep{lrpn,phs10}}};
    \node[below left=2.5cm and .5cm of lcs,text width=3.4cm](flcs){Functional lossy channel
                     systems {\citep{cwsts1}}};
    \node[below right=1cm and 1.5cm of aff](fm){Affine counter
                     systems with finite monoid {\citep{fm}}};
    \node[below right=1cm and 2.7cm of fm](broad){Broadcast
                     protocols {\citep{emerson,broadcast}}};
    \node[below=1cm of broad](tr){Transfer Petri nets {\citep{rtransnets}}};
    \node[below=1.5cm of fm](reset){Reset Petri
                     nets\\{\citep{rtransnets,boundedRPN}}};
    \node[below right=1.5cm and 1.5cm of reset](pn){Petri nets\\~};
    \path
      (aff) edge (minsky)
      (pres) edge (minsky)
      (lcs) edge (minsky)
      (flcs) edge[bend left] (lcs)
      (lcn) edge (lcs)
      (fm) edge (aff)
      (fm) edge (pres)
      (reset) edge (fm)
      (broad) edge (fm)
      (tr) edge (broad)
      (flcs) edge (reset)
      (lcn) edge (reset)
      (pn) edge (reset)
      (pn) edge (tr);
    \begin{pgfonlayer}{background}
    \draw[color=mydarkgreen!80,rounded corners,shade,top color=mydarkgreen!9,middle color=mydarkgreen!7,shading=axis,shading angle=-42,very thick]
    (flcs.south west) -- (flcs.north west)
    .. controls (-6,1) .. (aff.north east) .. controls
    (7.5,-3.8) and (tr.south east) .. (pn.south east);
    \draw[color=mydarkgreen,rounded corners,shade,top color=mydarkgreen!20,middle
    color=mydarkgreen!11,shading=axis,shading angle=-42,very thick]
    (flcs.south west) -- (flcs.north west) .. controls (-4,-2) and (-3,-4) .. (-1,-.75) .. controls (8,-3.2) and (tr.south east) .. (pn.south east) -- (pn.south west);
    \end{pgfonlayer}
    \node[color=mydarkgreen!90,font=\Large] at (-3,-1.2) {WSTS};
    \node[color=mydarkgreen,font=\Large] at (-2.3,-4.5) {$\infty$-effective cd-WSTS};
  \end{tikzpicture}
  \caption{\label{fig:nom}Classes of systems mentioned in
    the paper, with a few relevant references.}
\end{figure*}

\paragraph{Completeness}
\citet{cwsts2} define \emph{$\omega^2$-WSTS} as the class
of systems that can be completed, and provide an extensive off-the-shelf
algebra of datatypes with their completions~\citep{cwsts1}.  As they
argue, all the concrete classes of deterministic WSTS considered in
the literature are $\omega^2$.  Completed systems share their sets of
finite and infinite traces with the original systems: the added limit
states only influence transfinite sequences of transitions.

For instance, the whole class of \emph{affine counter systems}, with
affine transition functions of form $f(\vec x)=\vec A\vec x+\vec b$,
with $\vec A$ a $k\times k$ matrix of non negative integers and $\vec
b$ a vector of $k$ integers---encompassing reset/transfer Petri nets
and broadcast protocols---can be completed to configurations in
$(\mathbb{N}\cup\{\omega\})^k$.  Similarly, \emph{functional lossy
  channel systems}---a deterministic variant of lossy channel
systems~\citep[see also \autoref{sub:posta}]{cwsts1}---can work on
\emph{products}~\citep[Corollary~6.5]{fwlcs}.  On both accounts, the
completed functions are partial continuous.

\paragraph{Determinism}
Beyond deterministic systems, one can consider \emph{finite branching}
WSTS~\citep{cwsts1}.  These are defined as deterministic WSTS equipped
with a labeling function $\sigma$.  Consider a deterministic WSTS
$\tup{S,s_0,\mathcal{F},{\rightarrow},{\leq}}$, where $\mathcal{F}$ is
a finite alphabet of action names; together with a labeling $\sigma:
\mathcal{F}\rightarrow\Sigma$, it defines a possibly non deterministic
WSTS $\tup{S,s_0,\Sigma,{\rightarrow}',{\leq}}$ with $s\ru{a}'s'$ if
and only if there exists $f$ in $\mathcal{F}$ such that $s\ru{f}s'$
and $\sigma(f)=a$.

Assuming basic effectiveness assumptions on the so-called
\emph{principal filters} $\uc s$ of $S$, we can decide the following
sufficient condition for determinism on finite branching WSTS:

\begin{proposition}\label{propdet}
  Let $\mathcal{S}$ be defined by a deterministic WSTS
  $\tup{S,s_0,\mathcal{F},{\rightarrow},{\leq}}$ along with a labeling
  $\sigma: \mathcal{F}\rightarrow\Sigma$.  If finite bases can be
  computed for
  $\uc s\cap\uc s'$ for all $s,s'$ in $S$, and for
  $S$ itself,
  then one can decide whether, for all reachable states $s$ of $S$ and
  pairs $(f,f')$ of transition functions in $\mathcal{F}$ with
  $\sigma(f)=\sigma(f')$, $s\in\dom\ru{f}\cap\dom\ru{f'}$ implies
  $f=f'$.
\end{proposition}
\begin{proof}
  Let $B$ be a finite basis for $S$, i.e.\ $\uc B=S$, and let
  $D\eqdef\dom\ru{f}\cap\dom\ru{f'}$.

  We can reformulate the existence of an $s$ violating the condition
  of the proposition as a coverability problem, by checking whether $s_0$
  belongs to $\mathsf{Pred}^\ast(D)$, which is decidable thanks to the
  usual backward reachability algorithm if we provide a finite basis
  for $D$.  To that end, we first compute $B_f$ and $B_{f'}$ two
  finite bases for ${\dom\ru{f}}=\uc\mathsf{Pred}_\mathcal{S}(\uc
  B,f)$ and ${\dom\ru{f'}}=\uc\mathsf{Pred}_\mathcal{S}(\uc B,f)$
  using the $\mathsf{Pred}$-effectiveness of $\mathcal{S}$.  We then
  compute a finite basis for
  \begin{equation}\label{eq-bff}
    D = \bigcup_{s_f\in B_f,s_{f'}\in B_{f'}}(\uc s_f \cap\uc s_{f'})
  \end{equation}
  using the computation of finite bases for intersections of principal
  filters.
\end{proof}
\noindent For instance, labeled functional lossy channel systems and
labeled affine counter systems fit \autoref{propdet}; also note that
determinism is known to be \textsc{ExpSpace}-complete for labeled
Petri nets~\cite{yen,Schmitz11}.

Another extension beyond cd-WSTS is possible: Call a system $\mathcal{S}$
\emph{essentially deterministic} if, analogously to the
\emph{essentially finite branching} systems of \citet{wqo}, for each
state $s$ and symbol $a$, there is a single maximal element inside
$\mathsf{Post}_{\mathcal{S}}(s,a)=\{s'\in S\mid s\ru{a}s'\}$, which we
can effectively compute.  Indeed, from $\mathcal{S}$ we can construct a
deterministic system $\mathcal{S}_d$ with transitions
$s\ru{a}\mathsf{max}(\mathsf{Post}_{\mathcal{S}}(s,a))$ defined whenever
$\mathsf{Post}_{\mathcal{S}}(s,a)$ is not empty, for all $s$ in $S$ and
$a$ in $\Sigma$.  Thanks to monotonicity, any string recognized from
some state in $\mathsf{Post}_{\mathcal{S}}(s,a)$ can also be recognized
from $\mathsf{max}(\mathsf{Post}_{\mathcal{S}}(s,a))$, which entails
$T(\mathcal{S})=T(\mathcal{S}_d)$.

Recall that though most of \emph{infinite branching WSTS} can be embedded into their finite branching WSTS completion \cite{BFGHM-lics15}, this completion has no reason to be uniformly bounded or deterministic.

Finally, one can try to devise trace- and cover-equivalent
deterministic semantics for systems with unbounded \emph{but finite} branching, like
\emph{functional lossy channel systems} \citep{cwsts1} for lossy
channel systems, or reset Petri nets for lossy Minsky machines.  From
a verification standpoint, the deterministic semantics is then
equivalent to the classical one.

\paragraph{Effectiveness}
All the concrete classes of WSTS we have mentioned are
$\mathsf{Pred}$-effective, and we assume this property from all our
systems from now on.  It also turns out that $\infty$-effective
systems abound, including once more (completed) affine counter systems
\citep{cwsts2} and functional lossy channel systems.

\section{Deciding Trace Boundedness}\label{sec:dec}
We present in this section two semi-algorithms, first for trace
boundedness, which relies on the decidability of language emptiness in
WSTS, and then for trace unboundedness, for which we show that a
finite witness can be found in cd-WSTS.  In fact, this second
semi-algorithm can be turned into a full-fledged algorithm when some
extra care is taken in the search for a witness.
\begin{theorem}\label{cordec}
  Trace boundedness is decidable for
  $\infty$-effective cd-WSTS.  If the trace set is
  bounded, then one can compute an adequate bounded expression
  $w_1^\ast\cdots w_n^\ast$ for it.
\end{theorem}
An additional remark is that \autoref{cordec} holds more generally for
the boundedness of the \emph{language} $L(\mathcal S)$ of a WSTS
instead of its trace set $T(\mathcal S)$.  Indeed, the semi-algorithm
for boundedness would work just as well with $L(\mathcal S)$,
while the semi-algorithm for unboundedness can restrict its search for
a witness to~$\mathsf{Pre}^\ast(F)$.
 
\subsection{Trace Boundedness}
Trace boundedness is semi decidable with a
rather straightforward procedure for any WSTS $\mathcal{S}$ (neither
completeness nor determinism are necessary): enumerate the possible bounded
expressions $w_1^\ast\cdots w_n^\ast$ and check whether the trace set
$T(\mathcal{S})$ of the WSTS is included in their language.
This last operation can be performed by checking the emptiness of the
language of the WSTS obtained as the synchronous product
$\mathcal{S}\times\mathcal{A}$ of the original
system with a DFA $\mathcal{A}$ for the complement of the language of
$w_1^\ast\cdots w_n^\ast$.  If $L(\mathcal{S}\times\mathcal{A})$ is
empty, which is decidable thanks to the generic backwards algorithm for
WSTS, then we have found a bounded expression for $T(\mathcal{S})$.

\subsection{Trace Unboundedness}\label{sub:dec}
We detail the procedure for trace unboundedness of the trace set.  Our
construction relies on the existence of a witness of trace
unboundedness, which can be found after a finite time in a cd-WSTS by
exploring its states using accelerated sequences.
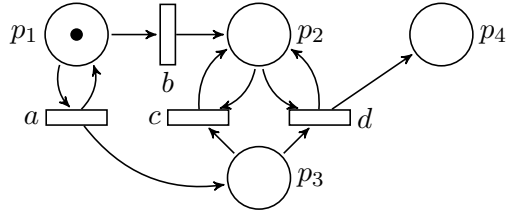
\begin{figure}[tb]
  \centering
  \begin{tikzpicture}[->,>=stealth',shorten >=1pt,initial text=,%
                    node distance=1.4cm,on grid,semithick,auto,
                    inner sep=2pt,every transition/.style={minimum width=8mm,minimum height=2mm}]
  \node[place,tokens=1,label=left:$p_1$](0){};
  \node[place,below right=1.9cm and 2.4cm of 0,label=right:$p_3$](2){};
  \node[transition,below=1.1cm of 0,label=left:$a$]{}
    edge[pre,bend left](0) edge[post,bend right](0) edge[post,bend right](2);
  \node[place,right=2.4cm of 0,label=right:$p_2$](1){};
  \node[place,right=2.4cm of 1,label=right:$p_4$](3){};
  \node[transition,right=1.2cm of 0,label=below:$b$,minimum width=2mm,minimum height=8mm]{}
    edge[pre](0) edge[post](1);
  \node[transition,below left=1.1cm and 0.8cm of 1,label=left:$c$]{}
    edge[pre](2) edge[pre,bend right](1) edge[post,bend left](1);
  \node[transition,below right=1.1cm and 0.8cm of 1,label=right:$d$]{}
    edge[pre](2) edge[pre,bend left](1) edge[post,bend right](1) edge[post](3);
  \end{tikzpicture}
  \caption{\label{fig:unb}The Petri net $\mathcal{N}'(1,0,0,0)$, with an unbounded trace set.}
\end{figure}

\paragraph{Overview}
Let us consider the Petri net $\mathcal{N}'$ with
initial marking $(1,0,0,0)$, depicted in \autoref{fig:unb}, with trace set
\begin{equation*}
  T(\mathcal{N}'(1,0,0,0))=a^\ast\cup\bigcup_{n\geq 0}a^nb\{c,d\}^{\leq n}\;.
\end{equation*}
Notice that the trace set of $\mathcal{N}'$ with initial marking
$(0,1,n,0)$ is bounded for each $n$: it is $\{c,d\}^{\leq n}$,
a finite language.  The trace unboundedness of
$\mathcal{N}'(1,0,0,0)$ originates in its ability to reach every $(0,1,n,0)$
marking after a sequence of $n$ transitions on $a$ followed by a
$b$ transition.

Consider now transitions $(1,0,0,0)\ru{a}(1,0,1,0)$ and
$(1,0,0,0)\ru{b}(0,1,0,0)$.  The two systems $\mathcal{N}'(1,0,1,0)$
and $\mathcal{N}'(0,1,0,0)$ are respectively trace unbounded and trace
bounded.  More generally, \autoref{lem-remove-letter} will show later
that, if $L\subseteq\Sigma^\ast$ is an unbounded language, then there
exists $a$ in $\Sigma$ such that $a^{-1}L$ is also unbounded.  By
repeated applications of, we can find words $w$ of any length $|w|=n$
such that $w^{-1}L$ is still unbounded: this is the case of $a^n$ in
our example.  This process continues to the infinite, but in a WSTS we
will eventually find two states $s_i\leq s_j$, met after $i<j$ steps
respectively.  Let $s_i\ru{u}s_j$; by monotonicity we can recognize
$u^\ast$ starting from $s_i$.  In a cd-WSTS, there is a
lub-accelerated state $s$ with $s_i\ru{u^\omega}s$ that represents the
effect of all these $u$ transitions;
here $(1,0,0,0)\ru{a^\omega}(1,0,\omega,0)$.  The interesting
point is that our lub-acceleration finds the correct residual trace set:
$T(\mathcal{N}'(1,0,\omega,0))=(a^\ast)^{-1}T(\mathcal{N}'(1,0,0,0))$.

Again, we can repeatedly remove accelerated strings from the prefixes of
our trace set and keep it unbounded.  However, due to the wqo, an
infinite succession of lub-accelerations allows us to nest some loops
after a finite number of steps.  Still with the same example, we reach
$(1,0,\omega,0)\ru{b}(0,1,\omega,0)$, and---thanks to the
lub-acceleration---the source of trace unboundedness is now
visible because both $(0,1,\omega,0)\ru{c}(0,1,\omega,0)$ and
$(0,1,\omega,0)\ru{d}(0,1,\omega,1)$ are increasing,
thus by monotonicity $T(\mathcal{N}'(0,1,\omega,0))=\{c,d\}^\ast$.  By
continuity, for each string $u$ in $\{c,d\}^\ast$, there exists $n$ in
$\mathbb{N}$ such that $a^{n}bu$ is an actual trace of
$\mathcal{N}'(1,0,0,0)$.

The same reasoning can be applied to the Petri net of
\autoref{fig:petri} with initial marking $(1,0,0,P,0)$ for $P\geq 2$.
As mentioned in \autoref{sec:prelim}, its trace set is unbounded, but
the trace set of $\mathcal{N}$ with initial marking
$(0,1,n,P,0)$ is bounded for each $n$, since it is a finite language.
We reach
$(1,0,0,P,0)\ru{g^\omega}(1,0,\omega,P,0)\ru{ci}(0,1,\omega,P-1,1)$
and see that both $(0,1,\omega,P-1,1)\ru{ei}(0,1,\omega,P-1,1)$ and
$(0,1,\omega,P-1,1)\ru{ieei}(0,1,\omega,P-1,1)$ are increasing,
thus by monotonicity $T(\mathcal{N}(0,1,\omega,P-1,1))$ contains
$\{ei,ieei\}^\ast$.  Here continuity comes into play to show that
these limit behaviors are reflected in the set of finite traces
of the system: in our example, for each string $u$ in $\{ei,ieei\}^\ast$,
there exists a finite $n$ in $\mathbb{N}$ such that $g^{n}ciu$ is an actual
trace of $\mathcal{N}(1,0,0,P,0)$.

\paragraph{Increasing Forks}
\begin{figure}[tb]
\centering
\begin{tikzpicture}[->,shorten >=1pt,initial text=,%
                    node distance=2cm,on grid,semithick,auto,
                    inner sep=2pt]
  \node(s0){$s_0$};
  \node[right=of s0](s){$s$};
  \node[above right=1cm and 2cm of s](sa){$s_a$};
  \node[below right=1cm and 2cm of s](sb){$s_b$};
  \path[every node/.style={font=\footnotesize}]
    (s0) edge[->>] (s)
    (s) edge[-,color=white] node[color=black]{$au$} (sa)
    (s) edge[->>,swap] node{$\leq$} (sa)
    (s) edge[-,color=white] node[color=black]{$bv$} (sb)
    (s) edge[swap] node{$\leq$} (sb);
\end{tikzpicture}
\caption{\label{fig:fork}An increasing fork witnesses trace unboundedness.}
\end{figure}
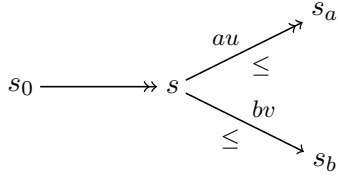
We call the previous witness of trace unboundedness an
\emph{increasing fork}, as depicted in schematic form in
\autoref{fig:fork}.  Let us first define accelerated runs and
languages for complete WSTS, where lub-accelerations are employed.
\begin{definition}
  Let $\mathcal{S}=\tup{S,s_0,\Sigma,\rightarrow,\leq,F}$ be a
  cd-WSTS.  An \emph{accelerated run} is a finite
  sequence $\sigma=s_0s_1s_2\cdots s_n$ in $S^\ast$ such that for all
  $i\geq 0$, either there exists $a$ in $\Sigma$ such that
\begin{align*}
  s_i&\ru{a}s_{i+1}\tag{single step}\\
  \intertext{or there exists $u$ in $\Sigma^+$ such that}
  s_i&\ru{u^\omega}s_{i+1}\;.\tag{accelerated step}
\end{align*}
  We denote the relation over $S$ defined by such an accelerated run by
  $s_0\rua{}s_n$.  An accelerated run is \emph{accepting} if $s_n$
  is in $F$.  The \emph{accelerated language} (resp.\ \emph{accelerated trace
  set}) $\Ltrans(\mathcal{S})$ (resp.\ $T_\mathsf{acc}(\mathcal{S})$) of
  $\mathcal{S}$ is the set of sequences that label some
  accepting accelerated run (resp.\ some accelerated run).%
\end{definition}%
\noindent
We denote by $\infwords$ the set of finite sequences mixing letters
$a$ from $\Sigma$ and accelerations $u^\omega$ where $u$ is a finite
sequence from $\Sigma^+$; in particular
\mbox{$\Ltrans(\mathcal{S})\subseteq\infwords$}.
\begin{definition}\label{def-incfork}
  A cd-WSTS
  $\mathcal{S}=\tup{S,s_0,\Sigma,\rightarrow,\leq}$ has an
  \emph{increasing fork} if there exist $a\neq b$ in $\Sigma$, 
  $u$ in $\infwords$, $v$ in $\Sigma^\ast$, and $s$, $s_a\geq s$,
  $s_b\geq s$ in $S$ such that $s_0\rua{} s$,  $s\rua{au}
  s_a$, and $s\ru{bv}s_b$.
\end{definition}

As shown in the following %
proposition, a semi-algorithm for
trace unboundedness in $\infty$-effective cd-WSTS then
consists in an exhaustive search for an increasing fork, by applying non
nested lub-accelerations whenever possible.
In fact, by choosing which acceleration sequences to employ in the search
for an increasing fork, we can turn this semi-algorithm into a full
algorithm; we will see this in more detail in \autoref{sub:ubound}.

\begin{proposition}\label{thrf}
  A cd-WSTS has an unbounded trace set if and only if it has an
  increasing fork.
\end{proposition}
The remainder of the section details the proof of
\autoref{thrf}.

\paragraph{An Increasing Fork Implies Unboundedness}
The following lemma shows that, thanks to continuity, what happens in
accelerated runs is mirrored in finite runs.  
\begin{lemma}\label{lem:finite}
  Let $\mathcal{S}$ be a cd-WSTS and $n\geq 0$.  If
  \begin{equation*}
    w_n=v_{n+1}u_n^\omega v_n\cdots u_1^\omega v_1\in T_\mathsf{acc}(\mathcal{S})
  \end{equation*}
  with the $u_i$ in $\Sigma^+$ and the $v_i$ in $\Sigma^\ast$, then %
  there exist $k_1,\ldots,k_n$ in $\nat$%
, such that 
  \begin{equation*}
    w'_n=v_{n+1}u_n^{k_n}v_n\cdots u_1^{k_1}v_{1}\in T(\mathcal{S})\;.
  \end{equation*}
\end{lemma}
\begin{proof}
  We proceed by induction on $n$.  In the base case where $n=0$,
  $w_0=v_{1}$ belongs trivially to $T(\mathcal{S})$---this concludes
  the proof if we are considering words in $T(\mathcal{S})$.  For the
  induction part, let $s$ be a state such that
  \begin{equation*}
    s_0\rua{v_{n+1}u_n^\omega}s\rua{v_nu_{n-1}^\omega v_{n-1}\cdots
  u_1^\omega v_1}s_f\;,
  \end{equation*}
  i.e.\ $w_{n-1}=v_nu_{n-1}^\omega v_{n-1}\cdots
  u_1^\omega v_1$ is in $T_\mathsf{acc}(\mathcal{S}(s))$.  Therefore, using the
  induction hypothesis, we can find $k_1,\ldots,k_{n-1}$ in $\nat$ %
  such that
  \begin{equation*}
    w'_{n-1}=v_nu_{n-1}^{k_{n-1}}v_{n-1}\cdots u_1^{k_1}v_1\in T(\mathcal{S}(s))\;.
  \end{equation*}
  Because $\mathcal{S}$ is complete, $\ru{w'_{n-1}}$ is a partial
  continuous function, hence with an open domain $O$.  This
  domain $O$ contains in particular $s$, which by definition of
  $u_n^\omega$ is the lub of the {directed set} $\{s'\mid\exists
  m\in\mathbb{N}, s_0\ru{v_{n+1}u_n^m}s'\}$.  By definition of an
  open set, there exists an element $s'$ in $\{s'\mid\exists
  m\in\mathbb{N}, s_0\ru{v_{n+1}u_n^m}s'\}\cap O$, i.e.\ there exists
  $k_n$ in $\mathbb{N}$ s.t.\ \mbox{$s_0\ru{v_{n+1}u_n^{k_n}}s'$} and
  $s'$ can fire the transition sequence~$w'_{n-1}$.
\end{proof}

Continuity is crucial for the soundness of our procedure, as can be
better understood by considering the example of the WSTS
$\mathcal{S}'=\tup{\mathbb{N}\uplus\{\omega\},0,\{a,b\},{\rightarrow},{\leq}}$
with transitions
\begin{align*}
  \forall n\in\mathbb{N},n\ru{a}&\;n+1,&
  \omega\ru{a}&\;\omega,&
  \omega\ru{b}&\;\omega\;.
\end{align*}
We obtain a bounded set of finite traces $T(\mathcal{S}'(0))=a^\ast$,
but reach the configuration $\omega$ through lub-accelerations, and then
find an increasing fork with $T(\mathcal{S}'(\omega))=\{a,b\}^\ast$, an
unbounded language.  Observe that $\mathbb{N}$ is a directed set with
$\omega$ as lub, thus the domain of $\ru{b}$ should contain some
elements of $\mathbb{N}$ in order to be open: $\mathcal{S}'$ is not
a complete WSTS.

\begin{lemma}\label{lem:fork-implies-unbounded}
  Let $\mathcal{S}$ be a cd-WSTS.  If $\mathcal{S}$ has an
  increasing fork, then $T(\mathcal{S})$ is unbounded.
\end{lemma}
\begin{proof}
  Suppose that $\mathcal{S}$ has an increasing fork with the same
  notations as in \autoref{def-incfork}, and let $w$ in
  $\infwords$ be such that $s_0\rua{w}s$.  By monotonicity,
  we can fire from $s$ the accelerated transitions of $au$ and the
  transitions of $bv$ in any order and any number of time, hence
  \begin{equation*}
    w \{au,bv\}^\ast\subseteq T_\mathsf{acc}(\mathcal{S})\;.
  \end{equation*}

  Suppose now that $T(\mathcal{S})$ is bounded, i.e.\ that there exists
  $w_1,\dots,w_n$ such that $T(\mathcal{S})\subseteq w_1^\ast\cdots
  w_n^\ast$.  Then, there exists a DFA
  $\mathcal{A}=\tup{Q,q_0,\Sigma,\delta,F}$ such that
  $L(\mathcal{A})=w_1^\ast\cdots w_n^\ast$ and thus
  $T(\mathcal{S})\subseteq L(\mathcal{A})$.  Set $N=|Q|+1$.  We have in
  particular
  \begin{equation*}
    w (bv)^Nau(bv)^Nau\cdots au(bv)^N\in T_\mathsf{acc}(\mathcal{S})
  \end{equation*}
  with $N$ repetitions of the $(bv)^N$ factor.  By
  \autoref{lem:finite}, we can find some adequate finite sequences
  $w',u_1,\dots,u_{N-1}$ in $\Sigma^\ast$ such that
  \begin{equation*}
    w'(bv)^Nau_1(bv)^Nau_2\cdots au_{N-1}(bv)^N\in T(\mathcal{S})\;.
  \end{equation*}

  Because $T(\mathcal{S})\subseteq L(\mathcal{A})$, this word is also
  accepted by $\mathcal{A}$, and we can find an accepting run for it.
  Since $N=|Q|+1$, for each of the $N$ occurrences of the $(bv)^N$
  factor, there exists a state $q_i$ in $Q$ such that
  $\delta(q_i,(bv)^{k_i})=q_i$ for some $k_i>0$.  Thus the accepting run
  in $\mathcal{A}$ is of form
  \begin{align*}
  q_0&\ru{w'(bv)^{N-k_1-k'_1}}q_1\ru{(bv)^{k_1}}q_1\ru{(bv)^{k'_1}au_1(bv)^{N-k_2-k'_2}}q_2,\\q_2&\ru{(bv)^{k_2}}q_2\ru{(bv)^{k'_2}au_2\cdots
    au_{N-1}(bv)^{N-k_N-k'_N}}q_N,\\q_N&\ru{(bv)^{k_N}}q_N\ru{(bv)^{k'_N}}q_f\in F
  \end{align*}
  for some integers $k'_i\geq 0$.  Again, since $N=|Q|+1$, there exist
  $1\leq i<j\leq N$ such that $q_i=q_j$, hence 
  \begin{equation*}
    \delta(q_i,(bv)^{k'_i}au_i\cdots au_{j-1}(bv)^{N-k_j-k'_j})=q_i\;.
  \end{equation*}
  This implies that $\{(bv)^{k_i},(bv)^{k'_i}au_i\cdots
  au_{j-1}(bv)^{N-k_j-k'_j}\}^\ast$ is contained in the set of factors of
  $L(\mathcal{A})$ with
  \begin{equation*}
    (bv)^{k_i+k'_i}au_i\cdots au_{j-1}(bv)^{N-k_j-k'_j}\neq
    (bv)^{k'_i}au_i\cdots au_{j-1}(bv)^{N-k_j-k'_j+k_i}
  \end{equation*}
  since $a\neq b$, thus $L(\mathcal{A})$ is an unbounded language
  \citep[\lemmaautorefname~5.3]{bcfl}, a contradiction.
\end{proof}

\paragraph{Unboundedness Implies an Increasing Fork}
We follow the arguments presented on the example of
\autoref{fig:petri}, and prove that an increasing fork can always be
found in an unbounded cd-WSTS.
\begin{figure}[tb]
  \centering
  \begin{tikzpicture}[->,shorten >=1pt,initial text=,%
      node distance=1.8cm,on grid,semithick,auto,
      inner sep=2pt]
  \node(s0){$s_0$};
  \node[right=1cm of s0](si){$s_i$};
  \node[right=1.2cm of si](s){$s$};
  \node[above right=0.9cm and 2cm of s](si1){$s_i$};
  \node[below right=0.9cm and 2cm of s](sp){$s_{i+1}$};
  \node[right=1.2cm of si1](sa){$s_b$};
  \node[right=2.2cm of sp](sj){$s_j$};
  \node[right=1.2cm of sj](sb){$s_a$};
  \path[every node/.style={font=\footnotesize}]
    (s0) edge[->>] node{} (si)
    (si) edge node{$x$} (s)
    (s)  edge node{$by$} (si1)
    (s)  edge[->>,swap] node{$azu_{i+1}^\omega$} (sp)
    (si1)edge node{$x$} (sa)
    (sp) edge[->>] node{$v_{i+2}\cdots u_j^\omega$} (sj)
    (sj) edge[swap] node{$x$} (sb)
    (si) edge[dashed,bend left=25] node{$=$} (si1)
    (s)  edge[dashed, bend right=10] node{$=$} (sa)
    (si) edge[swap,dashed, bend right=30] node{$\leq$} (sj)
    (s)  edge[dashed, bend left=12] node{$\leq$} (sb);
\end{tikzpicture}
\caption{\label{fig:thrf}The construction of an increasing fork in the
  proof of \autoref{lemma-exist-fork}.}
\end{figure}
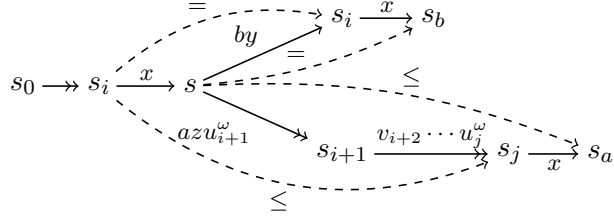

\begin{lemma}\label{lem-remove-letter}
  Let $L\subseteq\Sigma^\ast$ be an unbounded language. There exists
  $a$ in $\Sigma$ such that $a^{-1}L$ is also unbounded.
\end{lemma}
\begin{proof}
  Observe that $L = \bigcup_{a \in \Sigma}a\cdot(a^{-1}L)$. If every
  $a^{-1}L$ were bounded, since bounded languages are closed by finite
  union and concatenation, $L$ would also be bounded.
\end{proof}

\begin{definition}
  Let be $L\subseteq\Sigma^\ast$ and $w\in\Sigma^+$.  The
  \emph{removal} of $w$ from $L$ is the language $\remove w L =
  ((w^\ast)^{-1}L)\setminus w\Sigma^\ast$.
\end{definition}

\begin{lemma}\label{lem-remove-star}
  If a cd-WSTS $\mathcal{S}$ has an unbounded trace set
  $T(\mathcal{S})$ in $\Sigma^\ast$, and $L$ is an unbounded language
  with $L\subseteq T(\mathcal{S})$ then there are two words $v$ in
  $\Sigma^\ast$ and $u$ in $\Sigma^+$ such that $vu^\omega \in
  T_\mathsf{acc}(\mathcal{S}), vu \in \prefix{L}$ and
  $\remove{u}(v^{-1}L)$ is also unbounded.
\end{lemma}
\begin{proof}
  By \autoref{lem-remove-letter} we can find a sequence
  $(a_i)_{i>0} \in \Sigma^\omega$ such that for all $n$ in $\nat$,
  $(a_1\cdots a_n)^{-1} L$ is unbounded.  Let $(s_i)_{i\geq 0}$ be the
  corresponding sequence of configurations in $S^\omega$, such that
  $s_i\ru{a_{i+1}}s_{i+1}$.  Because $(S,\leq)$ is a wqo, there exist $i<j$
  such that $s_i \leq s_j$.  We set $v = a_1\cdots a_i$ and $u =
  a_{i+1}\cdots a_j$, which gives us $v\cdot u^\omega \in
  T_\mathsf{acc}(\mathcal{S})$.  Remark that $v^{-1}L$ is unbounded,
  and, since $u^*\remove{u}(v^{-1}L) = u^*(v^{-1}L)$,
  $\remove{u}(v^{-1}L)$ is unbounded too.
\end{proof}

Note that it is also possible to ask that $|vu| \geq n$ for any given
$n$, which we do in the proof of the following lemma.

\begin{lemma}\label{lemma-exist-fork}
  If a cd-WSTS has an unbounded
  trace set, then it has an increasing fork.
\end{lemma}
\begin{proof}
  We define simultaneously three infinite sequences, $(v_i,u_i)_{i>0}$ of
  pairs of words in $\Sigma^\ast\times\Sigma^+$, $(L_i)_{i\geq 0}$ of
  unbounded languages, and $(s_i)_{i\geq 0}$ of initial
  configurations: let $L_0\eqdef T(\mathcal{S})$ and $s_0$ the initial
  configuration of $\mathcal S$, and
  \begin{itemize}
  \item $v_{i+1},u_{i+1}$ are chosen using \autoref{lem-remove-star}
    s.t.\ $v_{i+1}u_{i+1}^\omega$ is in
    $T_\mathsf{acc}(\mathcal{S}(s_i))$, $v_{i+1}u_{i+1}$ is in $\prefix
    {L_i}$, $|v_{i+1}\cdot u_{i+1}| \geq |u_{i}|$ if $i>0$, and
    $\remove{u_{i+1}}(v^{-1}_{i+1}L_i)$ is unbounded;
  \item $s_i \rua{v_{i+1}u^\omega_{i+1}} s_{i+1}$;
  \item $L_{i+1}\eqdef\remove{u_{i+1}}(v^{-1}_{i+1}L_i)$.
  \end{itemize}
  Since $\remove{u_{i+1}}(v^{-1}_{i+1}L_i) \subseteq
  T(\mathcal{S}(s_{i+1}))$, we can effectively
  iterate the construction by the last point above.

  Due to the wqo, there exist $i<j$ such that $s_i \leq s_j$. By
  construction $u_{i}$ is not a prefix of $v_{i+1}u_{i+1}$ and
  $|v_{i+1}u_{i+1}| \geq |u_{i}|$, so there exist $a\neq b$ in $\Sigma$ and
  a longest common prefix $x$ in $\Sigma^\ast$ such that $u_{i} = xby$
  and $v_{i+1}u_{i+1} = xaz$ for some $y,z$ in $\Sigma^\ast$.

  We exhibit an increasing fork by selecting $s,s_a,s_b$ such that
  (see \autoref{fig:thrf}):
  \begin{align*}
  s_{i} \ru{x} s && s \rua{azu_{i+1}^\omega
  v_{i+2}u_{i+2}^\omega\cdots v_ju_j^\omega x} s_a && s \ru{byx} s_b\;.\tag*{\qedhere}
  \end{align*}
\end{proof}

We will refine the arguments of \autoref{lemma-exist-fork} in
\autoref{sub:ubound}.  In particular, n%
ote that a strategy where the
$v_{i+1}u_{i+1}$ sequences are the shortest possible defines a means
to perform an \emph{exhaustive} search for this particular brand of
increasing forks, this at no loss of generality as far as trace boundedness
is concerned.  Thus our semi-algorithm is actually an algorithm.%

\section{Undecidable Cases}\label{sec:undec}
This section establishes that the decidability of the trace
boundedness property for cd-WSTS disappears if we consider more
general systems or a more general property.  Unsurprisingly, trace
boundedness is undecidable on general systems like 2-counter Minsky
machines (\autoref{sub:nwsts}).  It also becomes undecidable if we
relax the determinism condition, as shown by considering the case of
labeled reset Petri nets (\autoref{sub:nondet}).  We conclude by
proving that post$^\ast$ flattability is undecidable for deterministic
WSTS (\autoref{sub:posta}).  Note that completeness is irrelevant in
all the following reductions.

\subsection{General Systems}\label{sub:nwsts}
We demonstrate that the trace boundedness problem is
undecidable for deterministic Minsky machines, by reduction
from their halting problem.  We could rely on Rice's Theorem, but find
it more enlightening to present a direct proof that turns a Minsky
machine $\mathcal{M}$ into a new one $\mathcal{M}'$, which halts if
and only if $\mathcal{M}$ halts.  The new machine has a bounded trace
set if it halts, and an unbounded trace set %
otherwise.

Let us first recall that a deterministic \emph{Minsky machine} is a
tuple $\mathcal{M}=\tup{Q,\delta,C,q_0}$ where $Q$ is a finite set of
labels, $\delta$ a finite set of actions, $C$ a
finite set of counters that take their values in $\mathbb{N}$, and
$q_0\in Q$ an initial label.  A label $q$ identifies a unique action
in $\delta$, which is of one of the following three forms:
\begin{align*}
    q:&\;\mathtt{if}\;c=0\;\mathtt{goto}\:q'\;\mathtt{else}\;c\text-\text-;\;\mathtt{goto}\:q''\\
    q:&\;c\text{++};\;\mathtt{goto}\:q'\\
    q:&\;\mathtt{halt}
\end{align*}
where $q'$ and $q''$ are labels and $c$ is a counter.  A configuration
of $\mathcal{M}$ is a pair $(q,m)$ with $q$ a label in $Q$ and
$m$ a marking in $\mathbb{N}^C$, and leads to a single next configuration
$(q',m')$ by applying the action labeled by $q$---which should be
self-explaining---if different from $\mathtt{halt}$.  A run of
$\mathcal{M}$ starts with configuration $(q_0,\vec{0})$ and halts if
it reaches a configuration that labels a $\mathtt{halt}$ action.  We
define the corresponding LTS semantics by $(q,m)\ru{q}(q',m')$ if
$(q,m)$ and $(q',m')$ are two successive configurations of
$\mathcal{M}$; note that there is at most one possible transition from
any $(q,m)$ configuration, thus this LTS is deterministic.  It is
undecidable whether a 2-counter Minsky machine halts \citep{minsky}.

We also need a small technical lemma that relates the size of bounded
expressions with the size of some special words.
\begin{definition}
  The \emph{size} of a bounded expression $w_1^\ast\cdots w_n^\ast$ is
  $\sum_{i=1}^n|w_i|$. 
\end{definition}
\begin{lemma}\label{lem:lowb}
  Let $v_m\in(\Sigma\uplus\Delta)^\ast$ be a word of form
  $u_1x_1u_2x_2u_3\dots u_mx_m$ with
  $m\in\mathbb{N}$, $u_i\in\Sigma^+$, $x_i\in\Delta^+$ and
  $|x_i|<|x_{i+1}|$ for all $i$.  If there exist
  $w_1,\dots,w_n$ in $(\Sigma\uplus\Delta)^\ast$ such that
  $v_m\in w_1^\ast\cdots w_n^\ast$, then $\sum_{i=1}^n|w_i|\geq m$. 
\end{lemma}
\newcommand{\alt}{\mathsf{alt}}
\begin{proof}
  We consider for this proof the number of alphabet alternations $\alt(w)$ of
  a word $w$ in $(\Sigma\uplus\Delta)^\ast$, which we define using the
  unique decomposition of $w$ as $y_1\cdots y_{\alt(w)}$ where each $y_i$
  factor is non empty and in an alphabet different from that of its 
  successor.  For instance, $\alt(v_m)$ is $2m$.  We relate the number of
  alternations produced by words $w_i$ of a bounded expression for $v_m$
  with their lengths.  More precisely, we show that, if
  \begin{equation*}
    v_m=w_1^{j_1}\cdots w_n^{j_n}\;,
  \end{equation*}
  then for all $1\leq i\leq n$
  \begin{equation}\label{eq:alt}
    \alt(w_i^{j_i})\leq 2|w_i|\;.
  \end{equation}
  Clearly, \eqref{eq:alt} holds if $|w_i|=0$ or $j_i=0$.  If a word $w_i$ is in
  $\Sigma^+$ or $\Delta^+$, then $\alt(w_i^j)=1$ for all $j>0$ and
  \eqref{eq:alt} holds again.  Otherwise, the word $w_i$ contains at
  least one alternation, and then $j_i\leq 2$:
  otherwise there would be two maximal $x$ factors (in $\Delta^+$) in
  $v_m$ with the same length.  As each alternation inside
  $w_i$ requires at least one more symbol, we verify \eqref{eq:alt}.
  Therefore,
  \begin{equation*}
    2m = \alt(w_1^{j_1}\cdots w_n^{j_n})\leq \sum_{i=1}^n\alt(w_i^{j_i})\leq
    2\sum_{i=1}^n|w_i|\:.\tag*{\qedhere}
  \end{equation*}
\end{proof}
\begin{proposition}\label{propminsky}
  Trace boundedness is undecidable for 2-counter Minsky machines.
\end{proposition}
\begin{proof}
  We reduce from the halting problem for a 2-counter Minsky machine
  $\mathcal{M}$ with initial counters at zero.  We construct a 4-counter
  Minksy machine $\mathcal{M}'$ such that $T(\mathcal{M}')$ is bounded
  if and only if $\mathcal{M}$ halts. 

  The machine $\mathcal{M}'$ adds two extra counters $c_3$ and $c_4$,
  initially set to zero, and new labels and actions to $\mathcal{M}$.
  These are used to insert longer and longer sequences of transitions at each
  step of the original machine: each label $q$ gives rise to the
  creation of five new labels $q',q'',q^\dagger,q^\ddagger,q^\flat$ that
  identify the following actions
  \begin{align*}
    q':&\;\mathtt{if}\;c_3=0\;\mathtt{goto}\:q^\dagger\;\mathtt{else}\;c_3\text-\text-;\;\mathtt{goto}\:q''\\
    q'':&\;c_4\text{++};\;\mathtt{goto}\:q'\\
    q^\dagger:&\;\mathtt{if}\;c_4=0\;\mathtt{goto}\:q^\flat\;\mathtt{else}\;c_4\text-\text-;\;\mathtt{goto}\:q^\ddagger\\
    q^\ddagger:&\;c_3\text{++};\;\mathtt{goto}\:q^\dagger\\
    q^\flat:&\;c_3\text{++};\;\mathtt{goto}\:q
  \end{align*}
  and each subinstruction $\mathtt{goto}\:q$ in the original actions is
  replaced by $\mathtt{goto}\:q'$.  The machine
  $\mathcal{M'}$ halts iff $\mathcal{M}$ halts.  If it 
  halts, then its trace set $T(\mathcal{M}')$ is a singleton $\{w\}$, and
  thus is bounded.  If it does not halt, then its trace set is the set of
  finite prefixes of an infinite trace of form
  \begin{align*}
    q_0(q'_1q''_1)^0q'_1u_1q_1(q'_2q''_2)^1&q'_2u_2q_2(q'_3q''_3)^2q'_3u_3\\&\cdots
    q_i(q'_{i+1}q''_{i+1})^iq'_{i+1}u_{i+1}q_{i+1}\cdots
  \end{align*}
  where $q_0q_1q_2\cdots q_iq_{i+1}\cdots$ is the corresponding trace of
  the execution of $\mathcal{M}$, and the $u_j$ are sequences in
  $\{q_j^\dagger,q_j^\ddagger,q_j^\flat\}^\ast$.  By
  \autoref{lem:lowb}, no expression $w_1^\ast\cdots w_n^\ast$ of
  finite size can be such that $T(\mathcal{M}')\subseteq w_1^\ast\cdots
  w_n^\ast$. 

  We then conclude thanks to the (classical) encoding of our 4-counter
  machine $\mathcal{M}'$ into a 2-counter machine $\mathcal{M}''$ using
  G\"odel numbers \citep{minsky}: indeed, the encoding preserves the
  trace set (un-)boundedness of~$\mathcal{M}'$.
\end{proof}

\subsection{Nondeterministic WSTS}\label{sub:nondet}
Regarding nondeterministic WSTS with uniformly bounded branching, we reduce state
boundedness for reset Petri nets, which is
undecidable~\citep[Theorem 13]{lrpn}, to trace boundedness for labeled reset
Petri nets.  From a reset Petri net we construct a labeled reset Petri net
similar to that of \autoref{fig:unb}, which hides the computation details
thanks to a relabeling of the transitions.  The new net consumes
tokens using two concurrent, differently labeled transitions, so that
the trace set can attest to state unboundedness.

Let us first recall that a marked \emph{Petri net} is a tuple
$\mathcal{N}=\tup{P,\Theta,f,m_0}$ where $P$ and $\Theta$ are finite sets of
places and transitions, $f$ a flow function from $(P\times
\Theta)\cup(\Theta\times P)$ to $\mathbb{N}$, and $m_0$ an
initial marking in $\mathbb{N}^P$.  The set of \emph{markings}
$\mathbb{N}^P$ is ordered component-wise by $m\leq m'$ iff $\forall p\in
P$, $m(p)\leq m'(p)$, and has the zero marking $\vec{0}$ as least
element, such that $\forall p\in P$, $\vec{0}(p) \eqdef 0$.  A transition
$t\in \Theta$ can be fired in a marking $m$ if $f(p,t)\geq m(p)$ for all
$p\in P$, and reaches a new marking $m'$ defined by $m'(p) \eqdef m(p) -
f(p,t) + f(t,p)$ for all $p\in P$.

A \emph{labeled} Petri net (without $\varepsilon$ labels) further
associates a labeling letter-to-letter homomorphism $\sigma:
\Theta\rightarrow\Sigma$, and can be seen as a finite branching WSTS
$\tup{\mathbb{N}^P,m_0,\Sigma,{\rightarrow},{\leq}}$ where
$m\ru{\sigma(t)}m'$ if the transition $t$ can be fired in $m$ and
reaches $m'$.  Determinism of such a system is decidable in
\textsc{ExpSpace}~\citep{yen}.  An important class of deterministic
Petri nets is defined by setting $\Sigma=\Theta$ and
$\sigma=\text{id}_\Theta$, thereby obtaining the so-called \emph{free
  labeled} Petri nets.

A \emph{reset} Petri net $\mathcal{N}=\tup{P,\Theta,R,f,m_0}$ is
a Petri net $\tup{P,\Theta,f,m_0}$ with a set $R\subseteq
P\times\Theta$ of reset arcs.  The marking $m'$ reached after a
transition $t$ from some marking $m$ is now defined for all $p$ in $P$
by
\begin{equation*}
m'(p)\eqdef\begin{cases}f(t,p)&\text{if }(p,t)\in R\\m(p) - f(p,t) +
f(t,p)&\text{otherwise.}\end{cases}
\end{equation*}

\begin{proposition}\label{propnondet}
  Trace boundedness is undecidable for labeled reset Petri nets.
\end{proposition}
\begin{proof}
  Let $\mathcal{N}=\tup{P,\Theta,R,f,m_0}$ be a reset Petri net.  We
  construct a $\sigma$-labeled reset Petri net $\mathcal{N}'$ which is
  bounded if and only if $\mathcal{N}$ is state bounded, thereby
  reducing the undecidable problem of state boundedness in reset Petri
  nets~\cite{boundedRPN}.

\begin{figure}[bt]
\centering
\begin{tikzpicture}[->,>=stealth',shorten >=1pt,initial text=,%
                    node distance=1.4cm,on grid,semithick,auto,
                    inner sep=2pt,every transition/.style={minimum
                    width=5mm,minimum height=10mm}]
  \node[draw,rounded corners,text width=4cm, text height=0.6cm,inner
                    sep=8pt,text badly
                    centered,color=black!80,dotted,fill=black!15]
                    (box) at (1.5,2.87) {$\mathcal{N}$};
  \node[place,tokens=1,label=below:$p_+$](pp){};
  \node[transition,above=2.87cm of pp,label=left:$a$](t){}
    edge[pre,bend right](pp) edge[post,bend left](pp);
  \node[place,right=2.8cm of t,label=right:$p$](p){};
  \node[place,right=2.8cm of pp,label=below:$p_-$](pm){};
  \node[transition,right=of pp,label=below:$b$](tm){$t_-$}
    edge[pre](pp) edge[post](pm);
  \node[transition,above left=1.4cm and 1cm of pm,label=left:$c$](tb){$t_p^c$}
    edge[pre](p) edge[pre,bend left](pm) edge[post,bend right](pm);
  \node[transition,above right=1.4cm and 1cm of pm,label=right:$d$](tc){$t_p^d$}
    edge[pre](p) edge[pre,bend right](pm) edge[post,bend left](pm);
\end{tikzpicture}
\caption{\label{fig:rpn}The labeled reset Petri net $\mathcal{N}'$ of the proof
of \autoref{propnondet}.}
\end{figure}
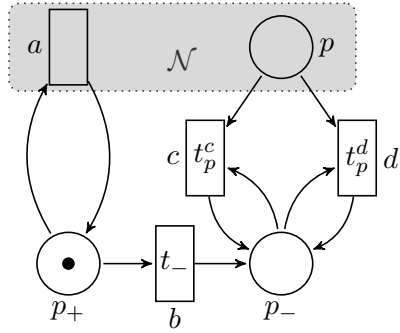
  We construct $\mathcal{N}'$ from $\mathcal{N}$ by adding two new
  places $p_+$ and $p_-$, two sets of new transitions $t^c_p$ and
  $t^d_p$ for each $p$ in $P$, where each $t^a_p$ for $a$ in $\{c,d\}$
  consumes one token from $p_-$ and from $p$ and puts one back in $p_-$,
  and one new transition $t_-$ that takes one token from $p_+$ and puts
  it in $p_-$.  All the transitions of $\mathcal{N}$ are modified to
  take one token from $p_+$ and put it back.  Finally, we set
  $m_0(p_+)=1$ and $m_0(p_-)=0$ in the new initial marking.  The
  labeling homomorphism $\sigma$ from
  $\Theta\uplus\{t_-\}\uplus\{t^a_p\mid a\in\{c,d\},p\in P\}$ to
  $\{a,b,c,d\}$ is defined by $\sigma(t)=a$ for all
  $t\in\Theta$, $\sigma(t_-)=b$, $\sigma(t^c_p)=c$ and $\sigma(t^d_p)=d$ for
  all $p$ in $P$.  See \autoref{fig:rpn} for a pictorial representation
  of $\mathcal{N}'$.  Its behavior is to simulate $\mathcal{N}$ while a
  token is in $p_+$ with $a^\ast$ for trace, and to switch
  nondeterministically to a consuming behavior when transferring this
  token to $p_-$ through $t_-$.  Then, $\mathcal{N}'$ consumes tokens
  from the places of $\mathcal{N}$ and produces strings in
  $\{c,d\}^\ast$ through the $t^c_p$ and $t^d_p$ transitions.

  If $\mathcal{N}$ is not state bounded, then $\sum_{p\in P}m(p)$ for
  reachable markings $m$ of $\mathcal{N}'$ is not bounded either.  Thus
  an arbitrary number of $t^c_p$ and $t^d_p$ transitions can be fired,
  resulting in a trace set containing any string in $\{c,d\}^\ast$ as
  suffix for $\mathcal{N}'$, which entails that it is not trace bounded.
  Conversely, if $\mathcal{N}$ is trace bounded, then $\sum_{p\in P}m(p)$ is
  bounded by some constant $n$ for all the reachable markings $m$ of
  $\mathcal{N}'$, hence $T(\mathcal{N}')$ is included in the set of
  prefixes of $a^\ast b\{c,d\}^n$, a bounded language. 
\end{proof}

\subsection{Trace vs.\ Post$^\ast$ Flattability}\label{sub:posta}
The decidability of trace boundedness calls for the investigation of the
decidability of less restrictive properties.  Two natural candidates
are post$^\ast$ flattability, which was proven undecidable for
Minsky machines by \citet{flataccel}, and cover flattability, which is
already known to be undecidable for cd-WSTS~\citep{cwsts2}.

We show that post$^\ast$ flattability is still undecidable for cd-WSTS.
To this end, we reduce again from state boundedness, this time in lossy
channel systems~\cite{lrpn}, to post$^\ast$ flattability in an
unlabeled \emph{functional} lossy channel system, a deterministic
variant introduced by \citet{cwsts1}.  Somewhat analogously to
\autoref{propnondet}, the idea is to consume the channel
contents on one end while adding an unbounded sequence to its other
end, so that the set of reachable configurations reveals state
unboundedness.

A \emph{lossy channel system} (LCS) is a WSTS
$\mathcal{C}=\tup{Q\times M^\ast,(q_0,\varepsilon),\{!,?\}\times
  M,{\rightarrow},{\preceq}}$ where $Q$ is a finite set of states,
$q_0\in Q$ the initial state, $M$ a finite set of messages,
$(q,w)\preceq(q',w')$ if $q=q'$ and $w\preceq w'$---the scattered
subword relation---, and where the transition relation is defined from
a finite relation $\delta\subseteq Q\times\{!,?\}\times M\times Q$
with
\begin{align*}
(q,w)&\ru{!a}(q',w')& \text{ if }(q,!,a,q')\in\delta\text{ and }\exists
  w''\in M^\ast,\\&& w''\preceq w\text{ and }w'\preceq w''a\\
(q,w)&\ru{?a}(q',w')& \text{ if }(q,?,a,q')\in\delta\text{ and }\exists
  w''\in M^\ast,\\&& aw''\preceq w\text{ and }w'\preceq w''\;.
\end{align*}
One can easily extend this definition to accommodate for a finite set of
channels and no-op transitions.

A \emph{functional lossy channel system}~\cite{cwsts1} is defined in the
same way except for the transition relation, which is now a partial
function:
\begin{align*}
(q,w)&\ru{!a}(q',wa)&\text{if }(q,!,a,q')\in\delta\\
(q,uaw)&\ru{?a}(q',w)&\!\!\!\!\!\!\!\text{if }(q,?,a,q')\in\delta\text{ and }u\in(M\setminus\{a\})^\ast\!.
\end{align*}
A functional LCS thus loses its channel contents lazily.
There are some immediate relations between a LCS $\mathcal{C}$ and its
corresponding functional LCS $\mathcal{C}'$, i.e.\ for the same $Q$, $M$,
and $\delta$:
$\mathsf{Post}^\ast_{\mathcal{C}'}((q_0,\varepsilon))\subseteq\mathsf{Post}^\ast_{\mathcal{C}}((q_0,\varepsilon))$,
$\mathsf{Cover}_{\mathcal{C}'}((q_0,\varepsilon))=\mathsf{Cover}_{\mathcal{C}}((q_0,\varepsilon))$,
and $T(\mathcal{C}')=T(\mathcal{C})$.

Note that the following proposition is \emph{not} a trivial
consequence of the undecidability of cover-flattability in LCS
\citep{cwsts2}, since in the case of functional LCS the
$\mathsf{Cover}$ and $\mathsf{Post}^\ast$ sets do not coincide.

\begin{proposition}\label{propreachlcs}
  Post$^\ast$ flattability is undecidable for functional lossy channel
  systems.
\end{proposition}
\begin{proof}Let us consider a LCS $\mathcal{C}=\tup{Q\times
M^\ast,(q_0,\varepsilon),\{!,?\}\times M,{\rightarrow},{\preceq}}$ and its
associated functional system $\mathcal{C}'$.  We construct a new
functional LCS $\mathcal{C}''$ which is post$^\ast$ flattable if and
only if $\mathcal{C}$ is state bounded, thereby reducing from the
undecidable state boundedness problem for lossy channel
systems~\cite{boundedRPN}.  Let us first remark that $\mathcal{C}$ is
state bounded if and only if $\mathcal{C}'$ is state bounded, if and
only if there is a maximal length $n$ to the channel content $w$ in any
reachable configuration
$(q,w)\in\mathsf{post}^\ast_{\mathcal{C}'}((q_0,\varepsilon))$.

We construct $\mathcal{C}''$ by adding two new states $q_?$ and $q_!$ to
$Q$, two new messages $c$ and $d$ to $M$, and a set of new
transitions to~$\delta$:
\begin{align*}
  &\{(q,?,a,q_!)\mid a\in M,q\in Q\}\\\cup\;&\{(q_!,!,a,q_?)\mid a\in\{c,d\}\}\\\cup\;&\{(q_?,?,a,q_!)\mid a\in M\}\;.
\end{align*}

If $\mathcal{C}'$ is state bounded, the writing transitions from $q_!$ can
only be fired up to $n$ times since they are interspersed with reading
transitions from $q_?$, hence $\mathcal{C}''$ has its channel
content lengths bounded by $n$.  Therefore, $\mathcal{C}''$ is
equivalent to a DFA with $(Q\uplus\{q_!,q_?\})\times
(M\uplus\{c,d\})^{\leq n}$ as state set and $\{!,?\}\times
(M\uplus\{c,d\})$ as alphabet.  By removing all the loops via a
depth-first traversal from the initial configuration
$(q_0,\varepsilon)$, we obtain a DFA $\mathcal{A}$ with a finite---and
thus bounded---language, but with the same set of reachable states.
Hence $\mathcal{C}''$ is post$^\ast$ flattable using~$\mathcal{A}$.

Conversely, if $\mathcal{C}'$ is not state bounded, then an arbitrarily
long channel content can be obtained in $\mathcal{C}''$, before
performing a transition to $q_!$ and producing an arbitrarily long
sequence in $\{c,d\}^\ast$ in the channel of $\mathcal{C}''$, witnessing
an unbounded trace suffix.  Observe that, due to the functional semantics,
$\mathcal{C}''$ has no means to remove these symbols, thus it has to put
them in the channel in the proper order, by firing the transitions from
$q_!$ in the same order.  Therefore no DFA with a bounded language can
be synchronized with $\mathcal{C}''$ and still allow all these
configurations to be reached: $\mathcal{C}''$ is not post$^\ast$
flattable.
\end{proof}

\section{Complexity of Trace Boundedness}\label{sec:cmplx}
Well-structured transition systems are a highly abstract class of
systems, for which no complexity upper bounds can be given in general.
Nevertheless, it is possible to provide precise bounds for several
concrete classes of WSTS, and even to employ generic proof techniques
to this end.  \autoref{tab:rrcmplx} sums up our complexity results,
using the \emph{fast-growing} complexity classes of \citep{schmitz13}.
\begin{table}[t]
  \caption{\label{tab:rrcmplx}Summary of complexity results for trace
    boundedness.} 
  \centering{\small
  \begin{tabular}{ccc}
    \toprule
    Petri nets & Affine counter systems & Functional LCS\\
    \midrule
    \textsc{ExpSpace}-complete & \textsc{Ack}-complete &
    \textsc{HAck}-complete\\
    \bottomrule
  \end{tabular}}
\end{table}

\subsection{Fast Growing Hierarchy}
Our complexity bounds are often adequately expressed in terms of a
family of fast growing functions, namely the generators
$(F_\alpha)_\alpha$ of the \emph{Fast Growing
  Hierarchy}~\citep{fastgr}, which form a hierarchy of ordinal-indexed
functions $\mathbb{N}\to\mathbb{N}$.  The first non primitive-recursive
function of the hierarchy is obtained for $\alpha=\omega$,
$F_\omega(n)=F_{n+1}(n)$ being a variant of the Ackermann function, and
eventually majorizes any primitive-recursive function.
Similarly, the first non multiply-recursive function is defined by
$\alpha=\omega^\omega$ and eventually majorizes any multiply-recursive
function.

Following \citep{schmitz13}, we define $\mathbf{F}_\alpha$
as the class of problems decidable using resources bounded by
$O(F_\alpha(p(n)))$ for instance size $n$ and some reasonable function
$p$ (formally, $p$ in $\bigcup_{\beta<\alpha}\mathscr{F}_{\beta}$
using the \emph{extended Grzegorczyk hierarchy}~\citep{fastgr}).
Since $F_3$ is already non elementary, the traditional distinctions
between space and time, or between deterministic computations and
nondeterministic ones, are irrelevant.  This gives rise to the
\emph{Ackermannian} complexity class
$\text{\textsc{Ack}}\eqdef\mathbf{F}_\omega$ and the
\emph{hyper-Ackermannian} complexity class $\text{\textsc{HAck}}\eqdef\mathbf{F}_{\omega^\omega}$.

\subsection{Lower Bounds}
Let us describe a generic recipe for establishing lower bounds: Given
a system $\mathcal{S}$ that simulates a space-bounded Turing machine
$\mathcal{M}$, hence with a finite number of different configurations
$n_c$, assemble a new system $\mathcal{S}'$ that first non
deterministically computes some $N$ up to $n_c$ (this is also known as
a ``weak'' computer for $n_c$), then simulates the runs of
$\mathcal{S}$ but decreases some counter holding $N$ at each
transition.  Thus $\mathcal{S}'$ terminates and has a bounded trace
set, but still simulates $\mathcal{M}$.  Now, add two loops on two
different symbols $a$ and $b$ from the configurations that simulate
the halting state of $\mathcal{M}$, and therefore obtain a system
which is trace bounded if and only if $\mathcal{M}$ does not halt.
Put differently, we reduce the control-state reachability problem in
terminating systems to the trace boundedness problem.

We instantiate this recipe in the cases of Petri nets in
\autoref{sub:exp}, using \citet{lipton76}'s results, for
reset Petri nets (and thus affine counter systems) in
\autoref{sub:rst} using \citet{acklcs}'s results, and for lossy
channel systems in \autoref{sub:lcs}, using \citet{CS-lics08}'s
results.  Although the complexity for Petri nets is quite
significantly lower than for the other classes of systems, we also
derive a non primitive-recursive lower bound on the \emph{size} of a
bounded expression for a trace bounded Petri net (\autoref{sub:prim}).  

\subsubsection{\textsc{ExpSpace}-Hardness for Petri Nets}\label{sub:exp}
Let us first observe that, since \citet{kmtree}-like constructions
always terminate in Petri nets, the search for an increasing fork is
an algorithm (instead of a semi-algorithm).  However, the complexity
of this algorithm is not primitive-recursive~\citep{cardoza}.

Meanwhile, we extend the \textsc{ExpSpace}-hardness result of
\citet{lipton76} for the Petri net coverability problem to the trace
boundedness problem.  As shown in~\citep{Schmitz11} using an extension
of the techniques of \citet{rackoff}, trace boundedness is in
\textsc{ExpSpace} for Petri nets, thus trace boundedness is
\textsc{ExpSpace}-complete for Petri nets.

\begin{proposition}\label{propexppn}
  Deciding the trace boundedness of a deterministic Petri net is
  \textsc{ExpSpace}-hard.
\end{proposition}
\begin{figure*}[t!]
\centering\hspace*{-1.7cm}
\begin{tikzpicture}[->,>=stealth',shorten >=1pt,initial text=,%
                    node distance=1.4cm,on grid,semithick,auto,
                    inner sep=2pt,every transition/.style={minimum width=8mm,minimum height=2mm}]
  \node[draw,rounded corners,text width=2cm, text height=3cm,inner
                    sep=10pt,text badly
                    centered,fill=black!15,draw=black!80,dotted]
                    (box) at (9.8,1.4) {$\mathcal{N}$};
  \node[place](P){$|P|$};
  \node[place,tokens=1,above right=2.8cm and 1.4cm of P](start){};
  \node[place,right=of P](end){};
  \node[transition,above right=1.4cm and 0.7cm of P](t1){}
    edge[pre](P)
    edge[pre](end)
    edge[post](start);
  \node[place,right=of end](n){$n$};
  \node[place,right=of start](save){};
  \node[transition,above right=1.4cm and 0.7cm of end](t2){}
    edge[pre](save)
    edge[pre,bend right](end)
    edge[post](n)
    edge[post,bend left](end);
  \node[place,right=of n](stop){};
  \node[place,right=of save](run){};
  \node[transition,above right=1.4cm and 0.7cm of n](t3){}
    edge[pre](n)
    edge[pre](stop)
    edge[post](save)
    edge[post](run);
  \node[transition,below=1.1cm of n,minimum width=2mm,minimum
                    height=8mm] (t4){}
    edge[pre,bend right](stop)
    edge[post,bend left](end);
  \node[transition,above=1.1cm of save,minimum width=2mm,minimum
                    height=8mm] (t5){}
    edge[pre,bend right](start)
    edge[post,bend left](run);
  \node[transition,above right=1.4cm and 0.7cm of stop](t6){}
    edge[pre](run)
    edge[post](stop);
  \node[place,right=of t6](mult){};
  \node[place,right=of mult,tokens=1,label=below:$p_t$](pt){};
  \node[transition,above=of mult](t7){}
    edge[pre,bend right](run)
    edge[pre](pt)
    edge[post,bend left](run)
    edge[post,swap]node{4}(mult);
  \node[transition,below=of mult](t8){}
    edge[pre](mult)
    edge[pre,bend right](stop)
    edge[post](pt)
    edge[post,bend left](stop);
  \node[transition,right=2.8cm of t7](t9){}
    edge[pre](pt);
  \node[transition,below=0.9cm of t9](t10){}
    edge[pre](pt);
  \node[below=0.9cm of t10](t11){\vdots};
  \node[transition,below=1cm of t11](t12){}
    edge[pre](pt);
  \node[place,right=of t9](p1){};
  \node[place,right=of t10](p2){};
  \node[right=of t11](p3){\vdots};
  \node[place,right=of t12](p4){};
  \node[right=4.2cm of pt](m){$m$};
  \path (p1) edge (m) (p2) edge (m) (p4) edge (m);
  \node[place,right=1.9cm of m,label=right:$p_h$](ph){};
  \node[transition,right=0.6cm of m,minimum width=2mm,minimum
                    height=8mm](tm){}
     edge[pre](m)
     edge[post](ph);
  \node[transition,above=of ph,label=above:$a$](ta){}
     edge[pre,bend left](ph) edge[post,bend right](ph);
  \node[transition,below=of ph,label=below:$b$](tb){}
     edge[pre,bend right](ph) edge[post,bend left](ph);
\end{tikzpicture}
\caption{\label{fig:lipton}The Petri net $\mathcal{N}'$ of the proof
of \autoref{propexppn}.}
\end{figure*}
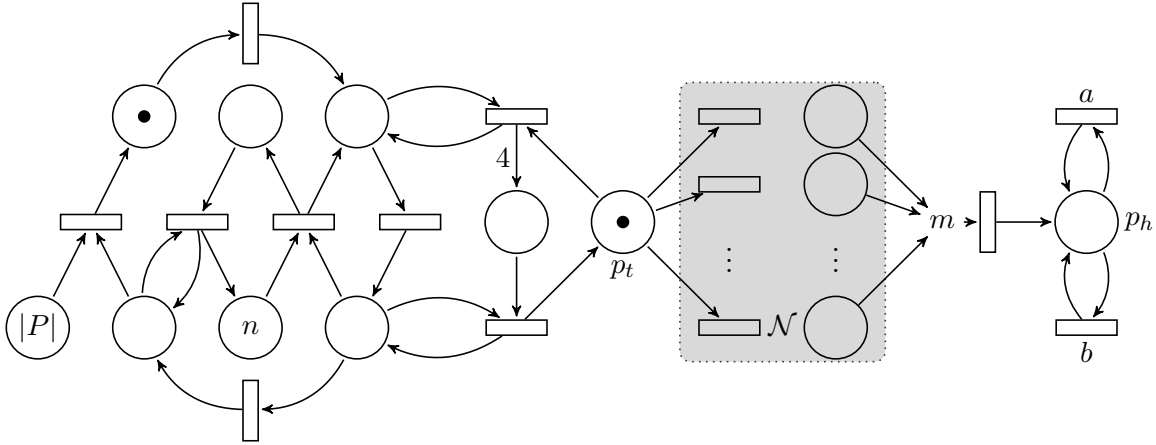
\begin{proof}%
  The \textsc{ExpSpace} hardness of deciding whether a
  Petri net has a bounded trace set can be shown by adapting a well-known
  construction by \citet{lipton76}---see also the description given by
  \citet{esparza}---for the \textsc{ExpSpace}-hardness of the
  coverability problem in Petri nets.  We refer the reader to their
  construction of an $O(n^2)$-sized $2^{2^n}$-bounded Petri net
  $\mathcal{N}$ that weakly simulates a $2^n$-space bounded Turing
  machine $\mathcal{M}$, such that a marking greater than some marking
  $m$ can be reached in $\mathcal{N}$ if and only if $\mathcal{M}$
  halts. 

  We construct a new free labeled Petri net $\mathcal{N}'$ from
  $\mathcal{N}=\tup{P,\Theta,f,m_0}$ and the marking $m$.  Since the places
  in $\mathcal{N}$ are bounded by $2^{2^n}$, only $n_c\eqdef2^{2^{n^{|P|}}}$
  different configurations are reachable from $m_0$ in $\mathcal{N}$,
  therefore we can limit the length of all the computations in $\mathcal{N}$
  to $n_c$ and still obtain the same reachability set. 

  We initially plug a subnet that ``weakly'' computes some $N\leq n_c$
  in a new place $p_t$ (displayed in the left part of
  \autoref{fig:lipton}), in less than $kn_c$ steps for some constant
  $k$.  This subnet only uses a constant size and an initial
  submarking of size $O(|P|+n)$.  We then simulate $\mathcal{N}$ but
  modify its transitions to consume one token from $p_t$ each time.
  Finally, a new transition that consumes $m$ from the subnet for
  $\mathcal{N}$ adds one token in another new place $p_h$ that allows
  two new different transitions $a$ and $b$ to be fired at will; see
  \autoref{fig:lipton}.
  
  A run of $\mathcal{N}'$ either reaches $p_h$ and can then have any
  string in $\{a,b\}^\ast$ as a suffix, or is of length bounded by
  $(k+1)n_c$.  Hence, $T(\mathcal{N}')$ is trace unbounded if and only
  if a run of $\mathcal{N}$ reaches some $m'\geq m$, if and only if
  the $2^n$-space bounded Turing machine $\mathcal{M}$ halts, which
  proves the \textsc{ExpSpace}-hardness of deciding the trace
  boundedness of a Petri net.
\end{proof}

\subsubsection{\textsc{Ack}-Hardness for Affine Counter Systems}\label{sub:rst}
\Citet{acklcs} shows that reset Petri nets (and thus affine counter
systems) can simulate Minsky
machines with counters bounded by
$F_k(x)$
for some finite $k$ and $x$.  Thus we can encode a
$F_\omega(n)$ space-bounded Turing machine using a $2^{F_\omega(n)}$-bounded
Minsky machine.  Since
\begin{equation*}
2^{F_\omega(n)}=2^{F_{n+1}(n)} \leq
F_2\!\left(F_{n+1}(n)\right)\leq F^2_{n+2}\!\left(n\right)\leq F_{n+3}\!\left(n+1\right)\,,
\end{equation*}
we can simulate this Minsky machine with a polynomial-sized reset
Petri net, and we get:
\begin{proposition}\label{prop:npr-acs}
  Trace boundedness of reset Petri nets is not primitive-recursive,
  more precisely it is hard for \textsc{Ack}.
\end{proposition}
\begin{proof}[Proof sketch]
  The construction is almost exactly the same as for the proof of
  \citet{acklcs}'s Theorem~7.1 of hardness of termination.  One simply
  has to replace extended instructions using reset transitions as
  explained in \citet{acklcs}'s Section~6, and to replace the single
  outgoing transition on $\ell_\omega$ by two different transitions,
  therefore yielding an unbounded trace set.
\end{proof}

\subsubsection{\textsc{HAck}-Hardness for Lossy Channel Systems}\label{sub:lcs}
\Citet{CS-lics08} show that LCS can weakly compute any
multiply-recursive function, and manage to simulate perfect
channel systems (i.e.\ Turing machines) of size bounded by such
functions, thereby obtaining a non multiply-recursive lower bound for
LCS reachability.  We prove that the same bound holds for
trace boundedness.

\newcommand{\egdef}{\stackrel{\mbox{\begin{scriptsize}def\end{scriptsize}}}{=}}
\newcommand{\obracew}[2]{{\overset{#2}{\overbrace{#1}}}}
\newcommand{\Lim}{{\mathit{Lim}}}
\newcommand{\FF}{{\mathfrak{F}}}
\newcommand{\Nat}{{\mathbb{N}}} %

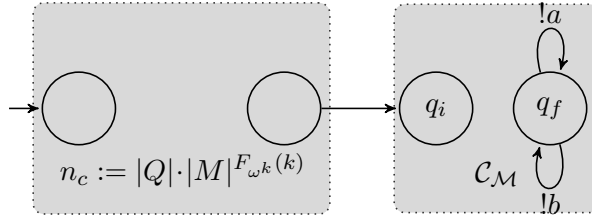
\begin{figure}[tb]
\centering
\begin{tikzpicture}[->,>=stealth',shorten >=1pt,initial text=,%
                    node distance=2cm,on grid,semithick,auto,
                    inner sep=2pt]
  \node[draw,rounded corners,text width=3.2cm, text height=2cm,inner
                    sep=10pt,text centered,fill=black!15,draw=black!80,dotted] at (2.35,0){$n_c:=|Q|\cdot|M|^{F_{\omega^k}(k)}$};
  \node[draw,rounded corners,text width=2cm, text height=2cm,inner
                    sep=10pt,text
                    centered,fill=black!15,draw=black!80,dotted] at (6.5,0){$\mathcal{C}_\mathcal{M}$};
  \node(s0){};
  \node[state,right=1cm of s0](s1){};
  \node[state,right=2.7cm of s1](s2){};
  \node[state,right=2cm of s2](s3){$q_i$};
  \node[state,right=1.5cm of s3](s4){$q_f$};
  \path
  (s0) edge node{} (s1)
  (s2) edge node{} (s3)
  (s4) edge[loop above] node{$!a$} ()
  (s4) edge[loop below] node{$!b$} ();
  
\end{tikzpicture}
\caption{\label{fig:LCShard} The lossy channel system
                    $\mathcal{C}_\mathcal{M}'$ for the proof of
                    Proposition~\ref{propnmrlcs}.}
\end{figure}
\begin{proposition}\label{propnmrlcs}%
  Trace boundedness of functional lossy channel systems is not
  multiply-recursive, more precisely it is hard for
  \textsc{HAck}.
\end{proposition}
\begin{proof}
  \Citet{CS-lics08} show that it is possible to perfectly simulate a
  Turing machine $\mathcal{M}$ with input $x$ and $k=|x|$ that works
  in space bounded by $F_{\omega^\omega}(k)=F_{\omega^{k+1}}(k)$, with an
  LCS $\mathcal{C}_\mathcal{M}$ of size polynomial in $k$ and
  $|\mathcal{M}|$, such that a state $q_f$ of
  $\mathcal{C}_{\mathcal{M}}$ is reachable if and only if
  $\mathcal{M}$ halts.  Furthermore, the number of distinct
  configurations $n_c=|Q|\cdot|M|^{F_{\omega^{k+1}}(k)}$ of
  $\mathcal{C}_\mathcal{M}$ can also be weakly computed in unary with
  an LCS of polynomial size, $Q$ being the set of states of
  $\mathcal{C}_\mathcal{M}$ and $M$ its message alphabet.

  Combining those two systems, we construct $\mathcal{C}_\mathcal{M}'$
  that
  \begin{enumerate}
  \item first ``weakly'' computes some $N\leq n_c$ (in a separate channel
    with a unary alphabet), and then
  \item executes $\mathcal{C}_\mathcal{M}$ while decrementing $N$ at
  each transition step, 
  \item is able to loop on two added transitions $q_f\ru{!a}q_f$ and
  $q_f\ru{!b}q_f$, which do not decrement $N$, giving rise to an
  unbounded trace.
  \end{enumerate}
  In a nutshell, all the runs of $\mathcal{C}_\mathcal{M}'$ that do not
  visit $q_f$ are terminating, being of length bounded by $n_c$.
  Consequently, $\mathcal{C}_\mathcal{M}'$ is unbounded if and only
  if $q_f$ is reachable, if and only if it was also reachable in
  $\mathcal{C}_\mathcal{M}$, if and only if $\mathcal{M}$ halts.

  We conclude the proof by remarking that both the weak computation of
  $n_c$ and the perfect simulation of $\mathcal{M}$ keep working with
  the functional lossy semantics.
\end{proof}

\subsubsection{Non Primitive-Recursive Size of a Bounded Expression for Petri Nets}\label{sub:prim}
We derive a non primitive-recursive lower bound on the computation of
the words $w_1,\dots,w_n$, already in the case of Petri nets.
Indeed, the size of a covering tree can be
non primitive-recursive compared to the size of the Petri net
\cite[who attribute the idea to Hack]{cardoza}.  Using the same insight,
we demonstrate that the words $w_1,\dots,w_n$ themselves can be of
non primitive-recursive size.  This complexity is thus inherent to the
computation of the~$w_i$'s.

\begin{proposition}\label{propnp}
  There exists a free labeled Petri net $\mathcal{N}$ with a bounded
  trace set $T(\mathcal{N})$ but such that for any words
  $w_1,\dots,w_n$, if $T(\mathcal{N})\subseteq w_1^\ast\cdots
  w_n^\ast$, then the size $\sum_{i=1}^n|w_i|$ is not
  primitive-recursive in the size of~$\mathcal{N}$.
\end{proposition}
\begin{proof}
We consider for this proof a Petri net that ``weakly'' computes a non
primitive-recursive function $A:\mathbb{N}\rightarrow\mathbb{N}$.  The
particular example displayed in \autoref{fig:Am} is taken from a survey by
\citet{j87}, where $A$ is defined for all $m$ and $n$~by
\begin{align*}
  A(n)&\eqdef A'_n(2) & A'_0(n) &\eqdef 2n+1 \\ A'_{m+1}(0)&\eqdef 1 & A'_{m+1}(n+1) &\eqdef A'_m(A'_{m+1}(n))\;.
\end{align*}%
\begin{figure*}[tb]
\centering
\begin{tikzpicture}[->,>=stealth',shorten >=1pt,initial text=,%
                    node distance=1.4cm,on grid,semithick,auto,
                    inner sep=2pt,every transition/.style={minimum width=2mm,minimum height=8mm}]
  \node[draw,rounded corners,text width=3.8cm, text height=2.8cm,inner
                    sep=10pt,text badly centered,fill=black!15,draw=black!80,dotted] (box)
       at (2.8,-1) {$A'_0$};
  \node[place,label=above:in$_0$] (in){};
  \node[place,below right=1.4cm and 1.4cm of in] (p0) {};
  \node[place,right=3cm of in] (p1) {};
  \node[place,right=3cm of p0] (p2) {};
  \node[place,label=below:out$_0$,right=2.6cm of p1] (out) {};
  \node[place,label=below:on$_0$,below=2.3cm of p0] (on) {};
  \node[place,label=below:off$_0$,below=2.3cm of p2] (off) {};
  \node[transition,right=of in] (t0){}
    edge[pre] (in)
    edge[pre,bend right] (p0)
    edge[post] node[auto]{2} (p1)
    edge[post,bend left] (p0);
  \node[transition,right=of p0] (t1){}
    edge[pre] (p0)
    edge[post,bend right](p1)
    edge[post] (p2);
  \node[fill,transition,right=of p1,label=above:$b$] (t2){}
    edge[pre] (p1)
    edge[pre,bend right] (p2)
    edge[post] (out)
    edge[post,bend left] (p2);
  \node[fill,transition,below=1.1cm of p0,label=left:$a$,minimum width=8mm,minimum height=2mm] (t3){}
    edge[pre](on)
    edge[post](p0);
  \node[transition,below=1.1cm of p2,minimum width=8mm,minimum height=2mm] (t4){}
    edge[post](off)
    edge[pre](p2);
\end{tikzpicture}\quad
\begin{tikzpicture}[->,>=stealth',shorten >=1pt,initial text=,%
                    node distance=1.4cm,on grid,semithick,auto,
                    inner sep=2pt,every transition/.style={minimum width=2mm,minimum height=8mm}]
  \node[draw,rounded corners,text width=7.2cm, text height=3.3cm,inner
    sep=10pt,text badly ragged,fill=black!15,draw=black!80,dotted]
    (box) at (3.45,-.75) {$A'_{m+1}$};
  \node[place,label=above:in$_m$] (in){};
  \node[place,below right=1.4cm and 1.4cm of in,label=left:on$_m$] (p0) {};
  \node[place,right=3cm of p0,label=below right:off$_m$] (p2) {};
  \node[place,right=3cm of p2,label=below:\mbox{$\;\;\:\text{out}_m=\text{out}_{m+1}$}]
     (out) {};
  \node[transition,rounded corners,text width=2.8cm,text
    height=.45cm,inner sep=5pt,right=2.9cm of in,text badly
    centered,fill=black!25,draw=black!80,dotted](An){$A'_m$}
    edge[pre] (in)
    edge[pre] (p0)
    edge[post,bend left](out)
    edge[post] (p2);
  \node[transition,right=of p2]{}
    edge[pre](out)
    edge[pre,bend left](p2)
    edge[post,bend right](p2)
    edge[post,bend right=65](in);

  \node[place,label=below:on$_{m+1}$,below=2.3cm of p0] (on) {};
  \node[place,label=below:off$_{m+1}$,below=2.3cm of p2] (off) {};
  \node[transition,below=1.1cm of p0,minimum width=8mm,minimum
    height=2mm]{}
    edge[pre](on)
    edge[post](p0);
  \node[transition,below=1.1cm of p2,minimum width=8mm,minimum
    height=2mm]{}
    edge[pre](p2)
    edge[post](off);

  \node[place,label=above:in$_{m+1}$,below left=1.4 and 1 of
  in](inn){};
  \node[transition, right=of p0]{}
    edge[pre](p2)
    edge[pre,bend left=30](inn)
    edge[post](p0);
\end{tikzpicture}
\caption{\label{fig:Am}A Petri net that ``weakly'' computes $A'_{m+1}$ \citep{j87}.}
\end{figure*}
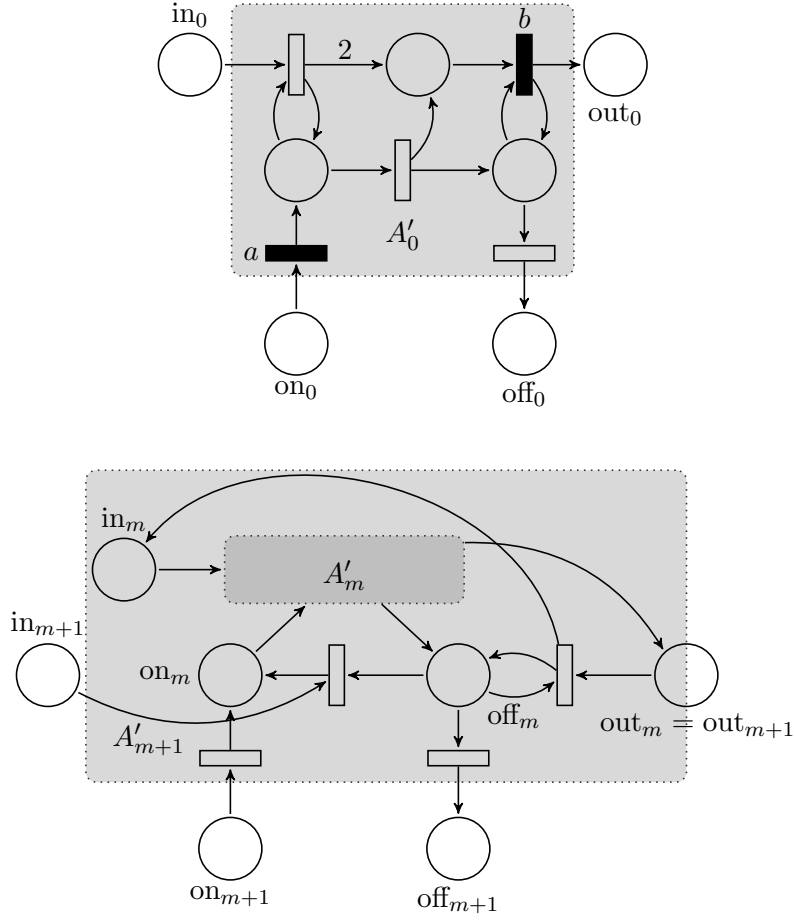%
The marked Petri net $\mathcal{N}$ for $A(n)$ is of linear size in $n$
and its trace set $L$ is finite, and therefore bounded, but contains words
of non primitive-recursive length compared to~$n$.

Although it might seem intuitively clear that we need a collection of
words $w_1,\dots,w_n$ of non primitive-recursive size in order to
capture this trace set, the proof is slightly more involved.  Observe for
instance that the finite trace set $\{a^p\}$ where $p$ is an arbitrary
number is included in the bounded expression $a^\ast$ of size
$|a|=1$.  Thus there is no general upper bound to the ratio between the size
$\sum_{w\in L}|w|$ of a finite trace set $L$ and the size of the minimal
collection of words that proves that $L$ is bounded.

Let us consider the maximal run in the Petri net for $A(n)$.  We focus
on the two black transitions labeled $a$ and $b$ in \autoref{fig:Am}, and more
precisely on the suffix of the run where we compute
$A'_n(2)=A'_1(p)$ with
\begin{equation*}
  p=A'_2(A'_3(\dots (A'_n(1) - 1)\dots)-1)\;.
\end{equation*}
This computation takes place in the subnet for $A'_0$ and $A'_1$ solely,
and this suffix is of form $v=ab^{k_0}ab^{k_1}\cdots ab^{k_p}$
with $k_0=1$, $k_{i+1}=2k_i+1$, and $k_p=A'_1(p)=A'_n(2)$.  By
\autoref{lem:lowb} any bounded expression such that $v\in w_1^\ast\cdots
  w_n^\ast$ has size $\sum_{i=1}^n|w_i|>p$.

We conclude by noting (1)~that $p$ is already the image of $n$ by a non
primitive-recursive function, and (2)~that $v$ is the suffix of the
projection $u$ of a word in $T(\mathcal{N})$ on the alphabet
$\{a,b\}$: hence, if a bounded expression of primitive-recursive size with
$T(\mathcal{N})\subseteq w_1^\ast\cdots w_n^\ast$ existed, then
the projections $w'_i$ of the $w_i$ on $\{a,b\}$ would be such that
$|w'_i|\leq|w_i|$ and $u\in {w'}_1^\ast\cdots {w'}_n^\ast$, and would
yield an expression of primitive-recursive size for~$v$.
\end{proof}

In the case of Petri nets we are in a situation comparable to
that of context-free languages: trace boundedness is decidable with a
sensibly smaller complexity than the complexity of the size of the
corresponding bounded expression (see \citet{growthrate} for a
\textsc{PTime} algorithm for deciding trace boundedness of a context-free
grammar, and \citet{noteslre} for an example of an expression
exponentially larger than the grammar).

\subsection{Upper Bounds}\label{sub:ubound}
We provide another recipe, this time for proving upper bounds for
trace-boundedness in cd-WSTS, relying on existing \emph{length
  function theorems} on wqos, which prove upper bounds on the length of
\emph{controlled bad sequences}.

\paragraph{Controlled Good and Bad Sequences}
Let $(S,\leq)$ be a quasi order.  A sequence $s_0\cdots s_\ell$ in
$S^\ast$ is \emph{$r$-good} if there exist $0\leq i_0<i_1<\cdots<i_r\leq \ell$ with
$s_{i_j}\leq s_{i_{j+1}}$ for all $0\leq j<r$, and is \emph{$r$-bad}
otherwise.  In the case $r=1$, we say more simply that the sequence is
\emph{good} (resp.\ \emph{bad}).  The wqo condition thus
ensures that any infinite sequence is good.

Given a \emph{norm} function $\norm{.}\,S\to\mathbb{N}$ with $S_{\leq
  n}\eqdef\{s\in S\mid \norm{s}\leq n\}$ finite for every $n$, a
\emph{control} function $g{:}\,\mathbb{N}\to\mathbb{N}$, $g$ monotone
s.t.\ $g(x)>x$ for all $x$, and an \emph{initial norm} $n$ in
$\mathbb{N}$, a sequence $s_0\cdots s_\ell$ is \emph{controlled} by
$(\norm{.},g,n)$ if, for all $i$, $\norm{s_i}\leq g^i(n)$ the $i$th
iteration of $g$; in particular, $\norm{s_0}\leq n$ initially.

A cd-WSTS $\tup{S,s_0,\Sigma,{\rightarrow},{\leq}}$ is
\emph{(strongly) controlled} by $(\norm{.},g,n)$ if
\begin{enumerate}
\item $\norm{s_0}\leq n$,
\item for any single step $s\ru{a}s'$,
  $\norm{s'}\leq g(\norm{s})$, and 
\item for any accelerated step $s\rua{u^\omega}s'$, $\norm{s'}\leq
  g^{|u|}(\norm{s})$.
\end{enumerate}

Using these notions, and by a careful analysis of the proof of
\autoref{thrf}, we exhibit in \autoref{sub:good} a witness of
trace unboundedness under the form of a good $(\norm{.},g^2,n)$-controlled
sequence $s_0\cdots s_\ell$ of $S^\ast$ in a $(\norm{.},g,n)$-controlled
WSTS.  There is therefore a longest bad prefix to this witness, which
is still controlled.

The particular way of generating this sequence yields an algorithm,
since as a consequence of the wqo, the depth of exploration in the
search for this witness of trace unboundedness is finite, and we can
therefore replace the two semi-algorithms of \autoref{sec:dec} by a
single algorithm that performs an exhaustive search up to this depth.
Furthermore, we can apply \emph{length function theorems} to obtain
upper bounds on the maximal length of bad controlled sequences, and
thus on this depth; this is how the upper bounds of
\autoref{tab:rrcmplx} are obtained (see \autoref{sub:acs_ubound} and
\autoref{sub:lcs_ubound}).

\subsubsection{Extracting a Controlled Good Sequence}\label{sub:good}
Let us assume we are given a trace unbounded $(\norm{.},g,n)$-controlled
cd-WSTS $\mathcal{S}$, and let us consider the three infinite sequences
defined in the proof of \autoref{lemma-exist-fork}, namely
$(v_i,u_i)_{i>0}$ of pairs of words in $\Sigma^\ast\times\Sigma^+$,
$(L_i)_{i\geq 0}$ of unbounded languages, and $(s_i)_{i\geq 0}$ of
states starting with the initial state $s_0$.  By construction,
$(s_i)_{i\geq 0}$ is good; however, this sequence is not controlled by
a ``reasonable'' function in terms of $g$, because we use the wqo
argument at each step (when we employ \autoref{lem-remove-star} to
construct $s_{i+1}$ from $s_i$), hence the motivation for refining
this first sequence.  A solution is to also consider some of the
intermediate configurations along the transition sequence
$v_{i+1}u_{i+1}$ starting in $s_i$, so that the index of each state in
the new sequence better reflects how the state was obtained.

\begin{lemma}\label{lem:ubound}
  Let $\mathcal{S}=\tup{S,s_0,\Sigma,{\rightarrow},{\leq}}$ be a
  $(\norm{.},g,n)$-controlled cd-WSTS.  Then we can construct a
  specific $(\norm{.},g^2,n)$-controlled sequence which is good if and only if
  $\mathcal{S}$ is trace unbounded. 
\end{lemma}
\begin{proof}
As in the proof of \autoref{lemma-exist-fork}, we construct
inductively on $i$ the following three infinite sequences
$(v_i,u_i)_{i>0}$, $(L_i)_{i\geq 0}$ starting with
$L_0\eqdef T(\mathcal{S})$, and $(s_i)_{i\geq 0}$ starting with the initial
state $s_0$ of $\mathcal{S}$, such that
\begin{itemize}
\item $v_{i+1},u_{i+1}$ are chosen using \autoref{lem-remove-star}
  such that
  \begin{enumerate}
  \item $v_{i+1}u_{i+1}^\omega$ is in $T_\mathsf{acc}(\mathcal{S}(s_i))$,
  \item $v_{i+1}u_{i+1}$ is in $\prefix {L_i}$,
  \item $|v_{i+1}|\geq |u_{i}|$ if $i>0$ (and thus
    $|v_{i+1}u_{i+1}|\geq|u_i|$ as in the proof of
    \autoref{lemma-exist-fork}),
  \item $\remove{u_{i+1}}(v^{-1}_{i+1}L_i)$ is unbounded, and
  \item there do not exist two successive strict prefixes $p,p'$ of 
    $v_{i+1}u_{i+1}$ such that $|p|\geq |u_i|$ and
    $s_i\ru{p}s'_i\ru{p'}s''_i$ with $s'_i\leq s''_i$,
    i.e.\ $v_{i+1}u_{i+1}$ is the shortest choice for
    \autoref{lem-remove-star} and (1--4) above;
  \end{enumerate}
\item $s_i \rua{v_{i+1}u^\omega_{i+1}} s_{i+1}$;
\item $L_{i+1} \eqdef\remove{u_{i+1}}(v^{-1}_{i+1}L_i)$.
\end{itemize}
We define another sequence of states $(s_{i,j})_{i\geq
0,j\in J_i}$ by $s_i\ru{p_{i,j}}s_{i,j}$ with $p_{i,j}$ the prefix of length
$j$ of $v_{i+1}u_{i+1}$, where
\begin{align*}
  J_0&\eqdef\{0,\dots,|v_1u_1|-1\}\text{ and}\\
  J_i&\eqdef\{|u_i|,\dots,|v_{i+1}u_{i+1}|-1\}\text{ for $i>0$.}
\end{align*}
Because $|u_i|>0$ for each $i>0$, none of the $(s_i)_{i>0}$ appears in
the sequence $(s_{i,j})_{i\geq 0,j\in J_i}$.  Note that condition~(5)
on the choice of $v_{i+1}u_{i+1}$ ensures that, for each $i\geq 0$,
each factor $(s_{i,j})_{j\in J_i}$ is a bad sequence.

This infinite sequence of states $(s_{i,j})_{i\geq 0,j\in J_i}$ can be
constructed whenever we are given a trace unbounded cd-WSTS, and is
necessarily \emph{good} due to the wqo.  Our aim will be later to
bound the length of its longest bad prefix.  In order to do so, we
need to control this sequence:

\begin{claim}\label{cl:control}
  The sequence $(s_{i,j})_{i\geq 0,j\in J_i}$ is controlled by
  $(\norm{.},g^2,n)$.
\end{claim}
\begin{proof}
  Since $\mathcal{S}$ is $(\norm{.},g,n)$-controlled, we can control the
  accelerated transition sequence that led to a given $s_{i,j}$: first
  reach $s_i$, and then apply $j$ single step transitions.  Formally,
  put for all $i\geq 0$
  \begin{align*}
    k_0&\eqdef 0,&k_{i+1}&\eqdef k_i+|v_{i+1}|+|u_{i+1}|\;,
  \end{align*}
  where $|v_{i+1}|$ accounts for the single steps and $|u_{i+1}|$ for
  the accelerated step in $s_i\rua{v_{i+1}u^\omega_{i+1}}s_{i+1}$;
  then we have for all $i\geq 0$ and $j\in J_i$
  \begin{equation*}
    \norm{s_{i,j}}\leq g^{k_i + j}(n)\;.
  \end{equation*}

  We need to relate this norm with the index of each $s_{i,j}$ in
  the $(s_{i,j})_{i\geq 0,j\in J_i}$ sequence.  We define accordingly
  for all $i\geq 0$ and $j\in J_i$
  \begin{align*}
    \ell_{0,\min J_0}&\eqdef 0,&\ell_{i,j+1}&\eqdef \ell_{i,j}+1,&\ell_{i+1,\min
      J_{i+1}}&\eqdef\ell_{i,\min J_i}+|J_i|\;.
  \end{align*}
  In order to prove our claim, namely that
  \begin{equation*}
    \norm{s_{i,j}}\leq (g^2)^{\ell_{i,j}}(n)\;,
  \end{equation*}
  we show by induction on $(i,j)$ ordered lexicographically that 
  \begin{equation*}
    k_i+j \leq 2\cdot \ell_{i,j}\;.
  \end{equation*}
  The base case for $i=0$ and $j=\min J_0=0$ is immediate, since
  $k_i+j=0=2\cdot\ell_{0,0}$.  For the induction step on $j$,
  $k_i+j+1\leq 2\cdot\ell_{i,j}+1\leq 2\cdot\ell_{i,j+1}$, and for the
  induction step on~$i$,
  \begin{align*}
    k_{i+1}+\min J_{i+1}
    &=k_{i+1}+|u_{i+1}|\tag{by def.\ of $J_{i+1}$}\\
    &=k_i+2|u_{i+1}|+|v_{i+1}|\tag{by def.\ of $k_{i+1}$}\\
    &=k_i+|u_i|+2|u_{i+1}|+|v_{i+1}|-|u_i|\\
    &=k_i+\min J_i+2|u_{i+1}|+|v_{i+1}|-|u_i|\tag{by def.\ of $J_i$}\\
    &\leq 2\cdot\ell_{i,\min J_i}+2|u_{i+1}|+|v_{i+1}|-|u_i|\tag{by ind.\ hyp.}\\
    &\leq 2\cdot\ell_{i,\min
      J_i}+2|u_{i+1}|+2|v_{i+1}|-2|u_i|\tag{since $|v_{i+1}|\geq |u_i|$}\\
    &=2\cdot\ell_{i+1,\min J_{i+1}}\;.\tag{by def.\ of $\ell_{i+1,\min
        J_{i+1}}$}
  \end{align*}

  Thus by monotonicity of $g$,
  \begin{equation*}
    \norm{s_{i,j}}\leq g^{k_i+j}(n)\leq g^{2\cdot\ell_{i,j}}(n)\;.\qedhere
  \end{equation*}
\end{proof}

We also need to show that such a good sequence is a \emph{witness} for
trace unboundedness, which we obtain thanks to
\autoref{lem:fork-implies-unbounded} and the following claim:
\begin{claim}\label{cl:fork}
  If the sequence $(s_{i,j})_{i\geq 0,j\in J_i}$ is good, then
  $\mathcal{S}$ has an increasing fork.
\end{claim}
\begin{proof}
  Let $s_{i,j}$ and $s_{i',j'}$ be two elements of the sequence
  witnessing goodness, such that $s_{i,j}$ occurs before $s_{i',j'}$ and
  $s_{i,j}\leq s_{i',j'}$.
  Due to the constraints put on the choices of $v_{i+1}$ and
  $u_{i+1}$ for each $i$, we know that $i<i'$.  Similarly to the proof
  of \autoref{lemma-exist-fork}, there exists a longest common
  prefix $x$ in $\Sigma^\ast$ and two symbols $a\neq b$ in $\Sigma$
  such that $v_{i+2}u_{i+2}=xaz$ and $u_{i+1}=xby$ for some $y$ and
  $z$ in $\Sigma^\ast$.  Let us further call $p'_{i,j}$ the suffix of
  $v_{i+1}u_{i+1}$ such that $v_{i+1}u_{i+1}=p_{i,j}p'_{i,j}$,
  hence we get a fork by selecting $s$, $s_a$, and $s_b$
  with
\begin{figure}[tb]
  \hspace*{-3em}
  \begin{tikzpicture}[->,shorten >=1pt,initial text=,%
      node distance=1.8cm,on grid,semithick,auto,
      inner sep=2pt]
  \node(s0){$s_0$};
  \node[right=1cm of s0](si){$s_i$};
  \node[right=1.2cm of si](sij){$s_{i,j}$};
  \node[right=1.2cm of sij](sss){$s''_i$};
  \node[right=1.3cm of sss](sip){$s_{i+1}$};
  \node[right=1.2cm of sip](s){$s$};
  \node[above right=0.9cm and 2cm of s](si1){$s_{i+1}$};
  \node[below right=0.9cm and 2cm of s](sp){$s_{i+2}$};
  \node[right=1.2cm of si1](sa){$s_b$};
  \node[right=2.3cm of sp](sj){$s_{i'}$};
  \node[right=1.2cm of sj](sijp){$s_{i',j'}$};
  \node[right=2cm of sijp](sb){$s_a$};
  \path[every node/.style={font=\footnotesize}]
    (s0) edge[->>] node{} (si)
    (si) edge node{$p_{i,j}$} (sij)
    (sij) edge node{$p'_{i,j}$} (sss)
    (sss) edge[->>] node{$u_{i+1}^\omega$} (sip)
    (sip) edge node{$x$} (s)
    (s)  edge node{$by$} (si1)
    (s)  edge[->>,swap] node{$azu_{i+2}^\omega$} (sp)
    (si1)edge node{$x$} (sa)
    (sp) edge[->>] node{$v_{i+3}\cdots u_{i'}^\omega$} (sj)
    (sj) edge node{$p_{i',j'}$} (sijp)
    (sijp) edge[->>,swap] node{$p'_{i,j}u_{i+1}^\omega x$} (sb)
    (sip) edge[dashed,bend left=25] node{$=$} (si1)
    (s)  edge[dashed, bend right=10] node{$=$} (sa)
    (sij) edge[dashed, bend right=30] node{$\leq$} (sijp)
    (s)  edge[dashed, bend left=12] node{$\leq$} (sb);
\end{tikzpicture}
\caption{\label{fig:forklcs}The construction of an increasing fork in the
  proof of \autoref{cl:fork}.}
\end{figure}
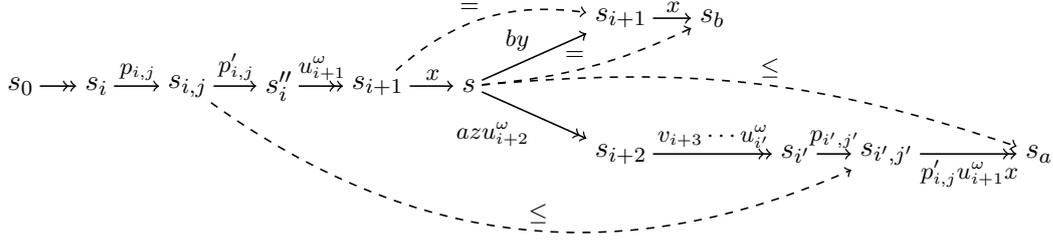
  \begin{align*}
    s_{i,j}\rua{p'_{i,j}u_{i+1}^\omega x}s&&s\rua{azu_{i+2}^\omega\cdots
      v_{i'}u_{i'}^\omega
      p_{i',j'}}s_{i',j'}\rua{p'_{i,j}u_{i+1}^\omega x}s_a&&s\ru{byx}s_b\;.
  \end{align*}
  Note that because $|x|<|u_{i+1}|$ and $i<i'$, $s_{i',j'}$ is
  necessarily met after $s$ and the construction is correct.  See also
  \autoref{fig:forklcs}.
\end{proof}
This concludes the proof of the lemma: $\mathcal{S}$ is trace
unbounded if and only if $(s_{i,j})_{i\geq 0,j\in J_i}$ is good.
\end{proof}

In the following, we essentially bound the complexity of trace
boundedness using bounds on the length of the $(s_{i,j})_{i\geq 0,j\in
  J_i}$ sequence.  This is correct modulo a few assumptions on the
concrete systems we consider, and because the fast growing upper bounds we
obtain dwarfen any additional complexity sources.  For instance, a
natural assumption would be for the size of representation of an
element $s$ of $S$ to be less than $\norm{s}$, but actually any
primitive-recursive function of $\norm{s}$ would still yield the same
upper bounds!

\subsubsection{$\mathbf{F}_\omega$ Upper Bound for Affine Counter
  Systems}\label{sub:acs_ubound}
We match the \textsc{Ack} lower bound of \autoref{prop:npr-acs}
for affine counter systems, thus establishing that trace boundedness is
\textsc{Ack}-complete.  We employ the machinery of
Claims~\ref{cl:control} and~\ref{cl:fork}, and proceed by showing that
\begin{enumerate}
\item\label{item:acs_control} complete affine counter systems are controlled,
  and that 
\item\label{item:n_k_ubound} one can provide an upper bound on the
  length of bad sequences in $(\mathbb{N}\uplus\{\omega\})^k$.
\end{enumerate}

\paragraph{Controlling Complete Affine Counter Systems}
Recall that an \emph{affine counter system} (ACS) $\tup{L,\vec x_0}$
is a finite set $L$ of affine transition functions of form $f(\vec
x)=\vec A\vec x+\vec b$, with $\vec A$ a matrix in
$\mathbb{N}^{k\times k}$ and $\vec b$ a vector in $\mathbb{Z}^k$,
along with an initial configuration $\vec x_0$ in $\mathbb{N}^k$. A
transition $f$ is firable in configuration $\vec x$ of $\mathbb{N}^k$
if $f(\vec x)\geq \vec 0$, and leads to a new configuration $f(\vec
x)$.

Define the \emph{norm} $\norm{\vec x}$ of a configuration in
$(\mathbb{N}\uplus\{\omega\})^k$ as the infinity norm among finite
values $\norm{\vec x}\eqdef\max(\{0\}\cup\{\vec x[j]\neq\omega\mid
1\leq j\leq k\})$.  Also set $m_1$ as the maximal coefficient
\begin{align*}
 m_1&\eqdef\max_{f(\vec x)=\vec A\vec x+\vec b\in L,1\leq i,j\leq k}\vec A[i,j]
\intertext{and $m_2$ as the maximal constant}
m_2&\eqdef\max_{f(\vec x)=\vec A\vec x+\vec b\in L,1\leq i\leq k}\vec b[i]\;.
\end{align*}

In case of a single step transition using some function $f$ in $L$,
one has
\begin{align*}
\norm{f(\vec x)}&\leq k\cdot{}m_1\cdot\norm{x}+m_2\;,
\end{align*}
while in case of an accelerated transition sequence, one has the
following:
\addtocounter{theorem}{1}
\setcounter{claim}{0}
\begin{claim}\label{cl:affine}
  Let $u=f_n\circ\cdots\circ f_1$ be a transition
  sequence in $L^+$ with $u(\vec x)\geq \vec x$.
  Then $\norm{(u^\omega(\vec x))}\leq (k\cdot m_1)^{n\cdot
    k}\cdot(\norm{\vec x}+n\cdot k\cdot m_2)$. 
\end{claim}
\begin{proof}
  We first proceed by proving that $k$ iterations of $u$ are enough in
  order to compute the finite values in $u^\omega(\vec x)$.

  Let us set $u(\vec x) = \vec A\vec x + \vec b$ and $\vec d_n\eqdef
  u^{n+1}(\vec x) - u^n(\vec x)$ for all $n$.  Since $u(\vec x) \geq
  \vec x$, for any coordinate $1\leq j\leq k$, the limit
  $\lim_{n\rightarrow \omega} u^n(\vec x)[j]$ exists, and is finite if
  and only if $\vec x[j]<\omega$ and there exists $m$ such that for
  all $n\geq m$, $\vec d_n[j] = 0$.  As $\vec d_{n+1} = u^{n+2}(\vec
  x) - u^{n+1}(\vec x) = \vec A \cdot u^{n+1}(\vec x) + \vec b - (\vec
  A \cdot u^n(\vec x) + \vec b) = \vec A \cdot (u^{n+1}(\vec
  x)-u^n(\vec x)) = \vec A \cdot \vec d_n$, we have $\vec d_n = \vec
  A^n\cdot \vec d_0$.

  If we consider $\vec A$ as the adjacency matrix of a weighted graph
  with $k$ vertices, its $\vec A^n[i,j]$ entry is the sum of the
  weights of all the paths $\psi=\psi_0\,\psi_1\cdots\psi_n$ of length
  $n$ through the matrix, which start from $\psi_0=i$ and end in
  $\psi_n=j$, i.e.
  \begin{align*}
  \vec A^n[i,j] &= \sum_{\psi\in \{i\} \times [1,k]^{n-1} \times \{j\}}
  \prod_{\ell\in [0,n-1]} \vec A[\psi_\ell,\psi_{\ell+1}]\\
  \vec d_n[j]&=\sum_{\psi\in\times [1,k]^{n} \times \{j\}} \left(\vec d_0[\psi_0]\cdot
  \prod_{\ell\in [0,n-1]} \vec A[\psi_\ell,\psi_{\ell+1}]\right)\,.
  \end{align*}

  Since $u(\vec x)\geq\vec x$ and $\vec A$ contains non negative
  integers from $\nat$, $\vec d_n[j] = \vec 0$ iff each of the above
  products is null, iff there is no path of length $n$ in the graph of
  $\vec A$ starting from a non-null $\vec d_0[i]$.  Therefore, if
  there exists $n > k$ such that $\vec d_n[j] > 0$, then there is a
  path with a loop of positive weight in the graph.  In such a case
  there are infinitely many $m$ such that $\vec d_m[j] > 0$.  A
  contrario, if there exists $m$ such that for all $n \geq m$, $\vec
  d_n[j] = 0$, then $m=k$ is enough: if $u^\omega(\vec x)[j] \in
  \nat$, then $u^\omega(\vec x)[j] = u^k(\vec x)[j]$.

  Let us now derive the desired upper bound on the norm of
  $u^\omega(\vec x)$: either $u^\omega(\vec x)[j]=\omega$ and the
  $j$th coordinate does not contribute to $\norm{u^\omega(\vec x)}$,
  or $u^\omega(\vec x)[j] \in \nat$ and $u^\omega(\vec x)[j] =
  u^k(\vec x)[j]$.  Let $f_i(\vec x)\eqdef\vec A_i\cdot\vec x+\vec
  b_i$; we have
  \begin{align*}
    \vec A &= \prod_{i=n}^1\vec A_i\qquad\qquad
    \vec b =\sum_{j=1}^{n}\left(\prod_{i=n}^{j+1} \vec
      A_i\right)\cdot \vec b_j\\
    u^k(\vec x)&= \vec A^k\cdot\vec x+\sum_{\ell=0}^{k-1}\vec
    A^\ell\cdot\vec b\\
    &=\left(\prod_{i=n}^1 \vec A_i\right)^{\!\!\!k}\cdot\vec
    x+\sum_{\ell=0}^{k-1}\sum_{j=1}^{n}\left(\prod_{i=n}^1 \vec
      A_i\right)^{\!\!\!\ell}\cdot\left(\prod_{i=n}^{j+1} \vec
      A_i\right)\cdot \vec b_j
    \intertext{thus}
    \norm{u^\omega(\vec x)[j]}
    &\leq\norm{\vec A^k\cdot\vec x}+\sum_{j=0}^{k-1}\norm{\vec A^j\cdot\vec b}\\
    &\leq(k\cdot m_1)^{n\cdot k}\cdot\norm{\vec x}+n\cdot k\cdot(k\cdot
    m_1)^{n\cdot k}\cdot m_2\\
    &=(k\cdot m_1)^{n\cdot k}\cdot(\norm{\vec x}+n\cdot k\cdot m_2)\qedhere
  \end{align*}
\end{proof}
\addtocounter{theorem}{-1}

\paragraph{Length Function Theorem}
It remains to apply the bounds of \citet{dickson} on the length of
controlled $r$-bad sequences over $\mathbb N^k$:
\begin{proposition}
  Trace boundedness for affine counter systems is in \textsc{Ack}.
\end{proposition}
\begin{proof}
  Define the projections $p_1$ and $p_2$ from
  $(\mathbb{N}\uplus\{\omega\})$ to $\mathbb{N}$ and $\{1,\omega\}$
  respectively by 
  \begin{align*}
    p_1(\omega)&\eqdef 0&p_1(n)&\eqdef n&
    p_2(\omega)&\eqdef\omega&p_2(n)&\eqdef 1
  \end{align*}
  for $n<\omega$, and their natural extensions from
  $(\mathbb{N}\uplus\{\omega\})^k$ to $\mathbb{N}^k$ and $\{1,\omega\}^k$.

  Consider the projection $(\vec x_{i,j})_{i\geq 0,j\in
    J_i}=\left(p_1(s_{i,j})\right)_{i\geq 0,j\in J_i}$ on
  $\mathbb{N}^k$ of the sequence defined in \autoref{sub:good}.  This
  sequence is $(\norm{.},g,\norm{\vec x_0})$-controlled if
  $(s_{i,j})_{i\geq 0,j\in J_i}$ is $(\norm{.},g,\norm{\vec
    x_0})$-controlled, and is $r$-good for any finite $r$ whenever the
  trace set of the affine counter system is unbounded.

  Conversely, if the sequence $(\vec x_{i,j})_{i\geq 0,j\in J_i}$ is
  $2^k$-good for the product ordering $\leq$ on $\mathbb{N}^k$, then
  the system has an increasing fork.  Indeed, let $r=2^k$; by
  definition of a $r$-good sequence, we can extract an increasing
  chain $\vec x_{k_0}\leq\vec x_{k_1}\leq\cdots\leq\vec x_{k_r}$ from
  the sequence $(\vec x_{i,j})_{i\geq 0,j\in J_i}$.  Since $r=2^k$, there
  exist $k_i<k_j$ such that $p_2(s_{k_i})=p_2(s_{k_j})$, and therefore
  $s_{k_i}\leq s_{k_j}$ and we can apply \autoref{cl:fork} to
  construct an increasing fork.

  By \autoref{cl:affine}, the sequence $(\vec x_{i,j})_{i\geq 0,j\in
    J_i}$ is $(\norm{.},g,\norm{\vec x_0})$-controlled by a
  primitive-recursive function $g$ (that depends on the size
  $\norm{L}$ of the affine counter system $\tup{L,\vec x_0}$ at hand),
  hence it is of length $\leq F_\omega(p(k,\norm{L},\norm{\vec
    x_0}))$ for some fixed primitive-recursive function
  $p$~\citep{dickson}.
\end{proof}

\subsubsection{$\mathbf{F}_{\omega^\omega}$ Upper Bound for Lossy Channel Systems}\label{sub:lcs_ubound}
\autoref{propnmrlcs} established a \textsc{HAck} lower
bound for the trace boundedness problem in lossy channel systems.  We match
this lower bound, thus establishing that trace boundedness is
\textsc{HAck}-complete.  As in \autoref{sub:acs_ubound}, we need
two results in order to instantiate our recipe for upper bounds: a
control on complete functional LCS, and a miniaturization for their
sequences of states.

\paragraph{Controlling Complete Functional LCS}
According to \citet{fwlcs}, LCS queue contents on an alphabet $M$
can be represented by \emph{simple regular expressions} (SRE) over
$M$, which are finite unions of \emph{products} over $M$.  %
Products, endowed with the language
inclusion ordering, suffice for the completion of functional LCS
\citep[Section~5]{cwsts1}, and thus for the representation of the
effect of accelerated sequences in functional LCS.

\setcounter{claim}{0}
\addtocounter{theorem}{2}
Products can be seen as finite sequences over a finite alphabet 
\begin{align*}
  \Pi_M=\{(a+\varepsilon)\mid a\in M\}\cup\{A^\ast\mid
  A\subseteq M\}
\end{align*}
with $|\Pi_M|=2^{|M|}+|M|$, with associated languages
$L(a+\varepsilon)\eqdef\{a,\varepsilon\}$ and $L(A^\ast)\eqdef
A^\ast$.  We consider the scattered subword ordering $\subword$ on
$\Pi_M^\ast$, defined as usual by $a_1\cdots a_m\subword b_1\cdots
b_n$ if there exists a monotone injection
$f:\{1,\dots,m\}\rightarrow\{1,\dots,n\}$ such that, for all $1\leq
i\leq n$, $a_i=b_{f(i)}$.  The scattered subword ordering is
compatible with language inclusion, thus we can consider the subword
ordering instead of language inclusion in our completed functional
LCS:\footnote{We could first define a partial ordering $\leq$ on
  $\Pi_M$ such that $(a+\varepsilon)\leq A^\ast$ whenever $a\in A$,
  and $A^\ast\leq B^\ast$ whenever $A\subseteq B$.  The corresponding
  subword ordering (using $a_i\leq b_{f(i)}$ in its definition) would
  be equivalent to language inclusion, and result in \emph{shorter}
  bad sequences.}
\begin{claim}\label{cl:subprod}
  For all products $\pi$, $\pi'$ in $\Pi_M^\ast$, $\pi\subword\pi'$
  implies $L(\pi)\subseteq L(\pi')$.
\end{claim}
Let us fix for the remainder of this section an arbitrary complete
functional LCS $\mathcal{C}=\tup{Q\times
  (\Pi_M)^\ast,(q_0,\varepsilon),\{!,?\}\times
  M,{\rightarrow},{\leq}}$, where $\leq$ is defined on
configurations in $Q\times(\Pi_M)^\ast$ by $(q,\pi)\leq(q',\pi')$ if $q=q'$
and $L(\pi)\subseteq L(\pi')$.

\begin{claim}\label{cl:control_lcs}
  Functional LCS are controlled by $(\norm{.},g,0)$ with
  $\norm{q,\pi}\eqdef|\pi|$ and $g(x)\eqdef 2^{x+2}+x$.
\end{claim}
\begin{proof}
  The claim follows from the results of \citet{fwlcs} on SREs.  Let
  the current configuration be $s\eqdef(q,\pi)$.

  In the case of a single transition step $s\ru{a}s'$, a product grows
  by at most one atomic expression $(a+\varepsilon)$
  \citep[Lemma~6.1]{fwlcs}.

  In the case of an accelerated transition step $s\rua{u^\omega}s'$ on
  a sequence $u$, since $s\ru{u}s''$ with $s\leq s''$, we are in one
  of the first three subcases of the proof of Lemma~6.4 of
  \citet{fwlcs}: the first two subcases yield the addition of an atomic
  expression $A^\ast$, while the third subcase adds at most $|u|^{|\pi|+2}$
  atomic expressions of form $(a+\varepsilon)$.
\end{proof}

\paragraph{Length Function Theorem}
\Citet{SS-icalp2011} give an upper bound on the least $N$ such that any
$(\norm{.},g,n)$-controlled sequence $\sigma$ with $|\sigma|=N$ of elements in
$(\Sigma^\ast,\subword)$ is $r$-good.\addtocounter{theorem}{-2}
\begin{fact}[\citeay{SS-icalp2011}]\label{fact:hig}
  Let $g$ be a primitive-recursive unary function and
  $n$ in $\mathbb{N}$.  Then, if $\sigma$ is a
  $(\norm{.},g,n)$-controlled $r$-bad sequence of
  $(\Sigma^\ast,\subword)$, then $|\sigma|\leq
  F^r_{\omega^{|\Sigma|-1}}(p(n))$ for some primitive-recursive~$p$.
\end{fact}

\begin{proposition}
  Trace boundedness for functional LCS is in \textsc{HAck}.
\end{proposition}
\begin{proof}
  We consider the sequence of products $(\pi_{i,j})_{i\geq 0,j\in
    J_i}$ extracted from the sequence of configurations
  $(s_{i,j})_{i\geq 0,j\in J_i}$ defined in \autoref{sub:good}.  This
  sequence of configurations is $r$-good for any finite $r$ whenever
  the trace set of the LCS is unbounded.  Conversely, if the sequence
  $(\pi_{i,j})_{i\geq 0,j\in J_i}$ is $(|Q|+1)$-good for the subword
  ordering $\subword$, then $\mathcal{C}$ has an increasing fork.
  Indeed, let $r=|Q|+1$; by definition of an $r$-good sequence, we can
  extract an increasing chain
  $\pi_{k_0}\subword\pi_{k_1}\subword\cdots\subword\pi_{k_r}$ of
  length $|Q|+1$ from the sequence $(\pi_{i,j})_{i\geq 0,j\in J_i}$.
  By \autoref{cl:subprod}, this implies $L(\pi_{k_0})\subseteq
  L(\pi_{k_1})\subseteq\cdots\subseteq L(\pi_{k_r})$.  Since $r=|Q|$,
  there exist $k_i<k_j$ such that $s_{k_i}=(q,\pi_{k_i})$ and
  $s_{k_j}=(q,\pi_{k_j})$ for some $q$ in $Q$.  Thus $s_{k_i}\leq
  s_{k_j}$, and we can apply \autoref{cl:fork} to construct an
  increasing fork.

  As the sequence $(\pi_{i,j})_{i\geq 0,j\in J_i}$ is
  $(\norm{.},g,0)$-controlled by a primitive-recursive function
  according to \autoref{cl:control_lcs}, the length of the sequence
  $(s_{i,j})_{i\geq 0,j\in J_i}$ need not exceed
  $F^{|Q|}_{\omega^{|\Pi_M|-1}}(p(|Q|))$ for some primitive-recursive
  $p$ by \autoref{fact:hig}, thus the upper bound of is
  multiply-recursive, and we obtain the desired
  $\mathbf{F}_{\omega^\omega}$ upper bound.
\end{proof}

\section{Verifying Trace Bounded WSTS}\label{sec:live}
As already mentioned in the introduction, liveness is generally
undecidable for cd-WSTS.  We show in this section
that it becomes decidable for trace bounded systems obtained as the
product of a cd-WSTS $\mathcal{S}$ with a
deterministic Rabin automaton: we prove that it is decidable whether the
language of $\omega$-words of such a system is empty
(\autoref{sub:declive}) and apply it to the LTL model checking problem
(\autoref{sub:decltl}).  We conclude the section with a short survey on
decidability issues when model checking WSTS (\autoref{sub:beyltl}); but
first we emphasize again the interest of trace boundedness for forward
analysis techniques.

\subsection{Forward Analysis}\label{sub:cover}
Recall from the introduction that a forward analysis of the set
of reachable states in an infinite LTS typically relies on
\emph{acceleration techniques} \citep[see e.g.][]{flataccel}
applied to loops $w$ in $\Sigma^\ast$, provided one can effectively
compute the effect of $w^\ast$.  Computing the full reachability set
(resp.\ coverability set for cd-WSTS) using a sequence $w_1^\ast\cdots
w_n^\ast$ requires post$^\ast$ flattability (resp.\ cover
flattability); however, as seen with \autoref{propreachlcs}
\citep[resp.][Proposition~6]{cwsts2}, both these properties are
already undecidable for cd-WSTS.

Trace bounded systems answer this issue since we can compute an
appropriate finite sequence $w_1$, \dots, $w_n$ and use it as
acceleration sequence.  Thus forward analysis techniques become
complete for trace bounded systems.  The Presburger accelerable counter
systems of \citet{foctlpr} are an example where, thanks to an
appropriate representation for reachable states, the full reachability
set is computable in the trace bounded case.
In a more WSTS-centric setting, the forward $\mathsf{Clover}$ procedure of
\citeauthor{cwsts2} for $\infty$-effective cd-WSTS terminates in the
cover flattable case \citep[Theorem~3]{cwsts2}, thus:
\begin{corollary}\label{cor:cover}
  Let $\mathcal{S}$ be a trace bounded $\infty$-effective cd-WSTS.  Then a
  finite representation of $\mathsf{Cover}_\mathcal{S}(s_0)$ can effectively
  be computed.
\end{corollary}
\noindent Using the $\mathsf{Cover}$ set, one can
answer state boundedness questions for WSTS.  Furthermore,
$\mathsf{Cover}$ sets and reachability sets coincide for lossy systems,
and lossy channel systems in particular.

\subsection{Deciding $\omega$-Language Emptiness}\label{sub:declive}
\paragraph{$\omega$-Regular Languages}
Let us recall the Rabin acceptance condition for $\omega$-words
(indeed, our restriction to deterministic systems demands a stronger
condition than the B\"uchi one).
Let us set some notation for infinite words in a labeled
transition system \mbox{$\mathcal{S}=\tup{S,s_0,\Sigma,\rightarrow}$}.  A
sequence of states $\sigma$ in $S^\omega$ is an \emph{infinite
  execution} for the infinite word $a_0a_1\cdots$ in $\Sigma^\omega$ if
$\sigma=s_0s_1\cdots$ with $s_i\ru{a_i} s_{i+1}$ for all $i$.   We
denote by $T_\omega(\mathcal{S})$ the set of infinite words that have
an execution.  The \emph{infinity set} of an infinite sequence
$\sigma=s_0s_1\cdots$ in $S^\omega$ is the set of
symbols that appear infinitely often in $\sigma$:
  $\mathsf{inf}(\sigma)=\{s\in S\mid|\{i\in\mathbb{N}\mid
s_i=s\}|=\omega\}$.%

Let $\mathcal{S}=\tup{S,s_0,\Sigma,\rightarrow,\leq}$ be a
deterministic WSTS and $\mathcal{A}=\tup{Q,q_0,\Sigma,\delta}$ a DFA.
A \emph{Rabin acceptance condition} is a finite
set of pairs $(E_i,F_i)_i$ of finite subsets of $Q$.  An infinite word $w$
in $\Sigma^\omega$ is accepted by $\mathcal{S}\times\mathcal{A}$ if
its infinite execution $\sigma$ over $(S\times Q)^\omega$ verifies
  $\bigvee_i(\mathsf{inf}(\sigma)\cap(S\times
  E_i)=\emptyset\wedge\mathsf{inf}(\sigma)\cap(S\times
  F_i)\neq\emptyset)$.  %
The set of accepted infinite words is denoted by
$L_\omega(\mathcal{S}\times\mathcal{A},(E_i,F_i)_i)$%
.  Thus
an infinite run is accepting if, for some $i$, it goes only finitely
often through the states of $E_i$, but infinitely often through the
states of~$F_i$.

\paragraph{Deciding Emptiness}%
We reduce the emptiness problem for
$L_\omega(\mathcal{S}\times\mathcal{A},(E_i,F_i)_i)$ to the
trace boundedness problem for a finite set of cd-WSTS,
which is decidable by \autoref{cordec}.
Remark that the following does not hold for nondeterministic
systems, since any system can be turned into a trace bounded
one by simply relabeling every transition with a single
letter~$a$.

\begin{theorem}\label{thlive}
  Let $\mathcal{S}$ be an $\infty$-effective cd-WSTS,
  $\mathcal{A}$ a DFA, and $(E_i,F_i)_i$ a Rabin condition.  If
  $\mathcal{S}\times\mathcal{A}$ is trace bounded, then it is decidable
  whether $L_\omega(\mathcal{S}\times\mathcal{A},(E_i,F_i)_i)$ is empty.
\end{theorem}
\begin{proof}
  Set $\mathcal{S}=\tup{S,s_0,\Sigma,\rightarrow,\leq}$ and
  $\mathcal{A}=\tup{Q,q_0,\Sigma,\delta}$.

  We first construct one cd-WSTS
  $\mathcal{S}_{i,1}$ for 
  each condition $(E_i,F_i)$ by adding to $\Sigma$ a fresh symbol $e_i$, to
  $S\times Q$ the pairs $(s,q_i)$ where $s$ is in $S$ and $q_i$ is a
  fresh state for each $q$ in $E_i$, and
  replace in $\rightarrow$ each transition $(s,q)\ru{a}(s',q')$ of
  $\mathcal{S}\times\mathcal{A}$ with $q$ in $E_i$ by two transitions
  $(s,q)\ru{e_i}(s,q_i)\ru{a}(s',q')$.  Thus we read in $\mathcal{S}_i$
  an $e_i$ marker each time we visit some state in~$E_i$.
  \begin{claim}\label{claimEi}%
    Each $\mathcal{S}_{i,1}$ is a trace bounded cd-WSTS.
  \end{claim}\begin{proof}[Proof of \autoref{claimEi}]
Observe that any trace of $\mathcal{S}_{i,1}$ is the image of a trace
of $\mathcal{S}\times\mathcal{A}$ by a \emph{generalized sequential
  machine} (GSM)
$\mathcal{T}_i=\tup{Q,q_0,\Sigma,\Sigma,\delta,\gamma}$ using $\Sigma$
both as input and output alphabet, and constructed
from $\mathcal{A}=\tup{Q,q_0,\Sigma,\delta}$ with the same set of states
and the same transitions, and by setting the output
function $\gamma$ from $Q\times\Sigma$ to $\Sigma^\ast$ to be
\begin{align*}
  (q,a)&\mapsto e_ia &\text{if }q\in E_i\\
  (q,a)&\mapsto a    &\text{otherwise.}
\end{align*}
A GSM behaves like a DFA on a word $a_1\cdots a_n$ by defining a run
$q_0\ru{a_1}q_1\cdots q_{n-1}\ru{a_n}q_n$ with
$q_{i+1}=\delta(q_i,a_{i+1})$ for all $i$, but additionally outputs
the word $\gamma(q_0,a_1)\gamma(q_1,a_2)\cdots\gamma(q_{n-1},a_n)$,
hence defining a function from finite words over its input alphabet to
finite words over its output alphabet.  Since bounded languages are
closed under GSM mappings \citep[Corollary on p.~348]{bcfl} and
$\mathcal{S}\times\mathcal{A}$ is trace bounded, we know that
$\mathcal{S}_{i,1}$ is trace bounded.
\end{proof}

  In a second phase, we add a new symbol $f_i$ and the elementary loops
  $(s,q)\ru{f_i}(s,q)$ for each $(s,q)$ in $S\times F_i$ to obtain a system
  $\mathcal{S}_{i,2}$.  Any run that visits some state in $F_i$ has
  therefore the opportunity to loop on~$f_i^\ast$.

  In $\mathcal{S}\times\mathcal{A}$, visiting $F_i$ infinitely often
  implies that we can find two configurations $(s,q)\leq(s',q)$ with
  $q$ in $F_i$.  In $\mathcal{S}_{i,2}$, we can thus recognize any
  sequence in $\{f_i,w\}^\ast$, where $(s,q)\ru{w}(s',q)$, from
  $(s',q)$: $\mathcal{S}_{i,2}$ is not trace bounded. 
  \begin{claim}\label{claimFi}%
    Each $\mathcal{S}_{i,2}$ is a cd-WSTS, and is trace
    unbounded iff there exists a run $\sigma$ in
    $\mathcal{S}\times\mathcal{A}$ with $\mathsf{inf}(\sigma)\cap(S\times
    F_i)\neq\emptyset$.
  \end{claim}
\begin{proof}[Proof of \autoref{claimFi}]
  If there exists a run $\sigma$ in $\mathcal{S}\times\mathcal{A}$
  with $\mathsf{inf}(\sigma)\cap(S\times F_i)\neq\emptyset$, then we
  can consider the infinite sequence of visited states in $S\times
  F_i$ along $\sigma$.  Since $\leq$ is a well quasi ordering on
  $S\times Q$, there exist two steps $(s,q)$ and later $(s',q')$ in
  this sequence with $(s,q)\leq(s',q')$.  Observe that
  the same execution $\sigma$, modulo the transitions introduced in
  $\mathcal{S}_{i,1}$, is also possible in $\mathcal{S}_{i,2}$.
  Denote by $w$ in $\Sigma^\ast$ the sequence of transitions between
  these two steps, i.e.\ $(s,q)\ru{w}(s',q')$.  By monotonicity
  of the transition relation of $\mathcal{S}_{i,2}$, we can recognize
  any sequence in $\{f_i,w\}^\ast$ from $(q',s')$.  Thus
  $\mathcal{S}_{i,2}$ is not trace bounded.

  Conversely, suppose that $\mathcal{S}_{i,2}$ is not trace bounded.  By
  \autoref{lemma-exist-fork}, it has an increasing fork with
  $(s_0,q_0)\mathrel{\mbox{$\rua{w}$}}(s,q)\rua{au}(s_a,q)$ and
  $(s,q)\ru{bv}(s_b,q)$, $s_a\geq s$, $s_b\geq s$, $a\neq b$ in
  $\Sigma\uplus\{e_i,f_i\}$, $u$, $w$ in
  $(\Sigma\uplus\{e_i,f_i\})_\mathrm{acc}$, and $v$ in
  $(\Sigma\uplus\{e_i,f_i\})^\ast$.

  Observe that
  if $f_i$ only appears in the initial segment labeled by $w$, then
  a similar fork could be found in $\mathcal{S}_{i,1}$, since $(s,q)$
  would also be accessible.  Thus, by
  \autoref{lem:fork-implies-unbounded}, $\mathcal{S}_{i,1}$ would not
  be trace bounded.  Therefore $f_i$ appears in $au$ or $bv$, and
  thus the corresponding runs for $au$ or $bv$ visit some state in
  $F_i$.  But then, by monotonicity, we can construct a run that visits
  a state in $F_i$ infinitely often.
\end{proof}

  In the last, third step, we construct the synchronous product
  $\mathcal{S}_{i,3}=\mathcal{S}_{i,2}\times\mathcal{A}_i$, where
  $\mathcal{A}_i$ is a DFA for the language $(\Sigma\uplus\{e_i\})^\ast
  f_i(\Sigma\uplus \{f_i\})^\ast$ (where $\uplus$ denotes a disjoint
  union).  This ensures that any run of $\mathcal{S}_{i,3}$ that goes
  through at least one $f_i$ cannot go through $e_i$ any longer, hence
  it visits the states in $E_i$ only finitely many often.  Since a run
  can always choose not to go through a $f_i$ loop, the previous claim
  still holds.  Therefore each $\mathcal{S}_{i,3}$ is a cd-WSTS, is
  trace unbounded iff there exists a run $\sigma$ in
  $\mathcal{S}\times\mathcal{A}$ with $\mathsf{inf}(\sigma)\cap(S\times
  E_i)=\emptyset$ and $\mathsf{inf}(\sigma)\cap(S\times
  F_i)\neq\emptyset$, and we can apply \autoref{cordec}.%
\end{proof}

\subsection{Model Checking LTL Formul\ae}\label{sub:decltl}
By standard automata-theoretic arguments~\cite{vw,safra}, one can
convert any linear-time temporal logic (LTL) formula $\varphi$ over a
finite set $\mathrm{AP}$ of atomic propositions, representing
transition predicates, into a deterministic Rabin
automaton $\mathcal{A}_{\neg\varphi}$ that recognizes exactly the runs over
$\Sigma=2^{\mathrm{AP}}$ that model $\neg\varphi$.  The synchronized product of
$\mathcal{A}_{\neg\varphi}$ with a complete, deterministic,
$\infty$-effective, and trace bounded WSTS $\mathcal{S}$ is again
trace bounded, and such that
$L_\omega(\mathcal{S}\times\mathcal{A},(E_i,F_i)_i)=T_\omega(\mathcal{S})\cap
L_\omega(\mathcal{A},(E_i,F_i)_i)$.
\autoref{thlive} entails that we can decide whether this language is empty,
and whether all the infinite traces of $\mathcal{S}$ verify
$\varphi$, noted $\mathcal{S}\models\varphi$.  This reduction also works
for LTL extensions that remain $\omega$-regular.
\begin{corollary}\label{corltl}
  Let $\mathcal{S}=\tup{S,s_0,2^{\mathrm{AP}},\rightarrow,\leq}$ be an
  $\infty$-effective trace bounded cd-WSTS, and $\varphi$ a LTL
  formula on the set $\mathrm{AP}$ of atomic propositions.  It is
  decidable whether $\mathcal{S}\models\varphi$. 
\end{corollary}

An alternative application of \autoref{thlive} is, rather than
relying on the trace boundedness of $\mathcal{S}$, to ensure that
$\mathcal{A}_{\neg\varphi}$ is trace bounded.  To this end, the following
slight adaptation of the flat counter logic of \citet{flatltl} is
appropriate:%
\begin{definition}
  A LTL formula on a set $\mathrm{AP}$ of atomic propositions is
  \emph{co-flat} if it is of form $\neg\varphi$, where $\varphi$
  follows the abstract syntax, where $a$ stands for a letter
  in~$2^\mathrm{AP}$:
  \begin{align}
    \varphi &::=
    \varphi\wedge\varphi\mid\varphi\vee\varphi\mid\mathsf{X}\varphi\mid\alpha\mathsf{U}\varphi\mid\mathsf{G}\alpha\tag{flat
    formul\ae}\\
    \alpha&::= \bigwedge_{p\in a}p\wedge\bigwedge_{p\not\in a}\neg
    p\;.\tag{alphabetic formul\ae}
  \end{align}
\end{definition}
\noindent In a conjunction $\varphi\wedge\varphi'$, one of $\varphi$ or
$\varphi'$ could actually be an arbitrary LTL formula.

One can easily check that flat formul\ae\ define languages of infinite
words with bounded sets of finite prefixes, and we obtain:
\begin{corollary}\label{corflatltl}%
  Let $\mathcal{S}=\tup{S,s_0,2^{\mathrm{AP}},\rightarrow,\leq}$ be an
  $\infty$-effective cd-WSTS, and $\varphi$ a
  co-flat LTL formula on the set $\mathrm{AP}$ of atomic propositions.
  It is decidable whether $\mathcal{S}\models\varphi$.
\end{corollary}
Extensions of \autoref{corflatltl} to less restrictive LTL fragments
seem possible, but our ideas thus far lead to rather unnatural
conditions on the shape of
formul\ae.%

\subsection{Beyond $\omega$-Regular Properties}\label{sub:beyltl}
We survey in this section some results from the model checking
literature and their consequences for several classes of trace bounded
WSTS.  Outside the realm of $\omega$-regular properties, we find
essentially two kinds of properties: state-based
properties or branching properties, or indeed a blend of the
two~\citep{foctlpr,lpar06,cwsts0}.

\paragraph{Affine Counter Systems}
Not all properties are decidable for trace bounded cd-WSTS, as seen with
the following theorem on affine counter systems.  Since these systems
are otherwise completable, deterministic, and $\infty$-effective,
action-based properties are decidable for them using
\autoref{thlive}, but we infer that state-based properties are
undecidable for trace bounded $\infty$-effective cd-WSTS.

\begin{theorem}[\citeay{racs}]%
  Reachability is undecidable for trace bounded affine counter systems.
\end{theorem}
Affine counter systems are thus the only class of systems
(besides Minsky counter machines) in \autoref{fig:nom} for which trace boundedness
does not yield a decidable reachability problem.

\paragraph{Presburger Accelerable Counter Systems}
\Citet{foctlpr} study the class of trace bounded counter systems for which
accelerations can be expressed as Presburger relations.\footnote{Whether
  trace boundedness is decidable for deterministic Presburger
  accelerable counter systems (i.e.\ not necessarily well-structured) is
  not currently known, while \autoref{propnondet} answers
  negatively in the nondeterministic well-structured case.}
Well-structured $\infty$-effective Presburger accelerable counter
systems include trace bounded reset/transfer Petri nets and broadcast
protocols, and \autoref{thlive} shows that $\omega$-regular
properties are decidable for them.

By the results of \citeauthor{foctlpr}, not only is the full
reachability set computable for these systems, but furthermore an
extension of state-based CTL$^\ast$ model checking with Presburger
quantification on the paths is also decidable.

\paragraph{Guarded Properties}
Let us recall that state-based LTL model checking is already undecidable
for Petri nets~\citep{ltlpn}.  However, state-based properties become
decidable for WSTS if they only allow to reason about upward-closed
sets.  This insight is applied by \citet{lpar06}, who define an upward and
downward guarded fragment of state-based $\mu$-calculus and prove its
decidability for all WSTS.  \Citet{cwsts0} presents a generalization to
open sets in well topological spaces.  Extensions of
\autoref{thlive} along these lines could be investigated.

\section{On Trace Unbounded WSTS}\label{sec:tunb}
As many systems display some commutative behavior, and on
that account fail to be trace bounded, \citet[Section~5.2]{flataccel}
introduce \emph{reductions} in order to enumerate the possible 
bounded expressions more efficiently, e.g.\ removal of identity loops, of
useless conjugated sequences of transitions, and of commuting
sequences.  Such reductions are systematically looked for, up to some
fixed length of the considered sequences.

Increasing forks suggest a different angle on this issue: whenever we
identify a source of trace unboundedness, we could try to check whether
the involved sequences commute, normalize our system, and restart the
procedure on the new system, which is trace-equivalent modulo the spotted
commutation.  Considering again the example Petri net of
\autoref{fig:unb}, the two sequences $c$ and $d$ responsible for an
increasing fork do commute.  If we were to force any sequence of
transitions in $\{c,d\}^\ast$ to be in the set $(cd)^\ast(c^\ast\cup
d^\ast)$, then 
\begin{itemize}
  \item the set of reachable states would remain the same, but
  \item the normalized trace set would be
    \begin{equation*}
      a^\ast\cup\bigcup_{0\leq 2m\leq n} a^nb(cd)^m(c^{\leq n-2m}\cup d^{\leq
        n-2m})\;,
    \end{equation*}
which is bounded.
\end{itemize}
Provided the properties to be tested do not depend on the relative
order of $c$ and $d$, we would now be able to apply
\autoref{thlive}.

We formalize this idea in
\autoref{sub:norm} using a partial commutation relation (see
\autoref{ax:com} for background on partial commutations), and
illustrate its interest for a bounded-session version of the
Alternating Bit Protocol (see \autoref{ax:abp} for background on this
protocol).

\newcommand{\cI}{\mathrel{\sim_{\!I}}}
\newcommand{\cIa}{\mathrel{\sim_{\!I}}}
\newcommand{\cIlim}{\mathrel{\sim_{\!I}^\mathrm{lim}}}
\newcommand{\clq}{\mathcal{C}}
\newcommand{\fnf}{\mathsf{fnf}_{\!I}}

\subsection{Partial Commutations}\label{ax:com}
Let $\Sigma$ be a finite alphabet; a \emph{dependence relation}
$D\subseteq\Sigma\times\Sigma$ is a reflexive and symmetric relation
on $\Sigma$.  Its complement $I=(\Sigma\times\Sigma)\setminus D$ is an
\emph{independence relation}.  On words in $\Sigma^\ast$, an
independence relation can be interpreted as a congruence
$\cI\;\subseteq\Sigma^\ast\times\Sigma^\ast$ generated by repeated
applications of $ab\mathrel{\leftrightarrow_I}ba$ for some $(a,b)$ in
$I$: $\cI\:=\:\leftrightarrow_I^\ast$, where
$w\mathrel{\leftrightarrow_I}w'$ if and only if there exist $u$ and
$v$ in $\Sigma^\ast$ and $(a,b)$ in $I$ with $w=uabv$ and $w'=ubav$.
We work on infinite words modulo the partial commutations described
by~$I$.

\paragraph{Closure}
The \emph{limit extension}
$\cIlim\;\subseteq\Sigma^\omega\times\Sigma^\omega$ of the congruence
$\cI$~\citep{regtraces,omcl} is defined by $\sigma\cIlim\sigma'$ iff,
\begin{itemize}
\item for every finite prefix $u$ of $\sigma$, there is a finite prefix
  $u'$ of $\sigma'$ and a finite word $v$ of $\Sigma^\ast$ such that
  $uv\cIa u'$, and
\item symmetrically, for every finite prefix $u'$ of $\sigma'$, there is
  a finite prefix $u$ of $\sigma$ and a finite word $v'$ of
  $\Sigma^\ast$ such that $u'v'\cIa u$.
\end{itemize}
Consider for instance the relation $I\eqdef\{(a,b),(b,a)\}$; then
$(aab)^\omega\cIlim(ab)^\omega$ (e.g.\ $(aab)^nb^n\cIa(ab)^{2n}$
and $(ab)^na^n\cIa(aab)^n$), but $(aab)^\omega\not\cIlim a^\omega$ (e.g.\ 
$(aab)^n v\not\cIa a^m$ for all $n>0$, $m>0$, and $v$
in $\Sigma^\ast$).

A language $L\subseteq\Sigma^\ast$ (resp.\ $L\subseteq\Sigma^\omega$)
is \emph{$I$-closed}, if for any $\sigma$ in $L$, and for every
$\sigma'$ with $\sigma\cIa\sigma'$ (resp.\ $\sigma\cIlim\sigma'$),
$\sigma'$ is also in $L$.
The closure of an $\omega$-regular language for a given partial
commutation is decidable, and more precisely \textsc{PSpace}-complete if
the language is given as a B\"uchi automaton or an LTL
formula~\citep{omcl}.

\begin{definition}
  An LTS is \emph{$I$-diamond} if, for any pair $(a,b)$
  of $I$, and for any states $s$ in $\dom\ru{ab}\cap\:\dom\ru{ba}$ and
  $s'$ in $S$, $s\ru{ab}s'$ iff $s\ru{ba}s'$.
\end{definition}
\noindent We have the following sufficient condition for the closure of
$T_\omega(\mathcal{S})$, which is decidable for $I$-diamond WSTS: just
compare the elements in the finite bases for $\dom\ru{ab}$ and
$\dom\ru{ba}$.
\begin{lemma}\label{lemwstscl}%
  Let $I$ be an independence relation and $\mathcal{S}$ an
  LTS, both on $\Sigma$.  If $\mathcal{S}$ is $I$-diamond and, for
  all $(a,b)$ of $I$, $\dom\ru{ab}\;=\dom\ru{ba}$, then
  $T_\omega(\mathcal{S})$ is $I$-closed.
\end{lemma}%
\begin{proof}
  One can easily check that this condition implies that the set of
  finite traces $T(\mathcal{S})$ is $I$-closed.

  Let now $\sigma$ be an infinite word in $T_\omega(\mathcal{S})$, and
  $\sigma'$ an infinite word in $\Sigma^\omega$ with
  $\sigma\cIlim\sigma'$, but suppose that $\sigma'$ is not in
  $T_\omega(\mathcal{S})$.  Thus there exists a finite prefix $u'$ of
  $\sigma'$ that does not belong to $T(\mathcal{S})$.  By definition of
  $\cIlim$, there is however a prefix $u$ of $\sigma$ and a word $v'$ of
  $\Sigma^\ast$ such that $u'v'\cIa u$.  But this contradicts the closure
  of $T(\mathcal{S})$, since $u$ is in $T(\mathcal{S})$, but $u'v'$ is
  not---or $u'$ would be in the prefix-closed language~$T(\mathcal{S})$.
\end{proof}

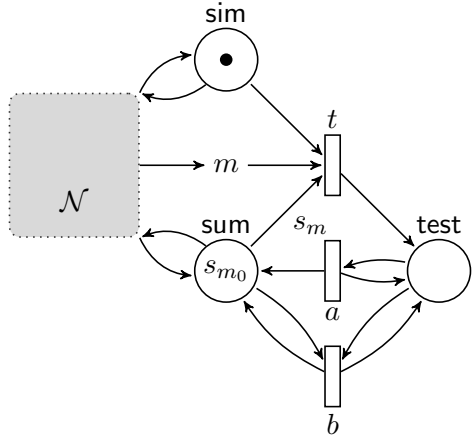
\begin{figure}[t]
  \centering
\begin{tikzpicture}[->,>=stealth',shorten >=1pt,initial text=,%
                    node distance=1.4cm,on grid,semithick,auto,
                    inner sep=2pt,every transition/.style={minimum width=2mm,minimum height=8mm}]
  \node[draw,rounded corners,text width=1cm, text height=1.2cm,inner
                    sep=10pt,text badly centered,fill=black!15,draw=black!80,dotted] (box){$\mathcal{N}$};
  \node[right=2cm of box](m){$m$};
  \node[place,tokens=1,above=of m,label=above:$\mathsf{sim}$](sim){};
  \node[place,below=of m,label=above:$\mathsf{sum}$](sum){$s_{m_0}$};
  \node[place,right=2.8cm of sum,label=above:$\mathsf{test}$](test){};
  \path
    (box.north east) edge[bend left]  (sim)
    (sim)            edge[bend left]  (box.north east)
    (box.east)       edge             (m)
    (box.south east) edge[bend right] (sum)
    (sum)            edge[bend right] (box.south east);
  \node[transition,right=of m,label=above:$t$](t){}
    edge[pre](sim) edge[pre](m) edge[pre]node{$s_m$}(sum) edge[post](test);
  \node[transition,below=of t,label=below:$a$](a){}
    edge[pre,bend left=15](test) edge[post,bend right=15](test) edge[post](sum);
  \node[transition,below=of a,label=below:$b$](b){}
    edge[pre,bend left=15](test) edge[pre,bend right=15](sum)
    edge[post,bend right=15](test) edge[post,bend left=15](sum);
\end{tikzpicture}
  \caption{\label{fig:trans}The transfer Petri net $\mathcal{N}'$ of the
  proof of \autoref{propcl}.}
\end{figure}
However, already in the case of $I$-diamond WSTS and already for
finite traces, $I$-closure is undecidable; a sufficient condition like
\autoref{lemwstscl} is the best we can hope for.
\begin{proposition}\label{propcl}%
  Let $I$ be an independence relation and $\mathcal{S}$ an $I$-diamond
  cd-WSTS, both on $\Sigma$.  It is undecidable whether
  $T(\mathcal{S})$ is $I$-closed or not.
\end{proposition}%
\begin{proof}
  We reduce the (undecidable) reachability problem for a transfer Petri net
  $\mathcal{N}$ and a marking $m$~\cite{rtransnets} to the $I$-closure
  problem for a new transfer Petri net $\mathcal{N}'$.  Let us recall
  that a \emph{transfer arc} $(p,t,p')$ transfers all the tokens from
  a place $p$ to another place $p'$ when $t$ is fired.

  The new transfer Petri net $\mathcal{N}'$ extends $\mathcal{N}$ with
  three new places $\mathsf{sim}$, $\mathsf{sum}$, and $\mathsf{test}$,
  and three new transitions $t$, $a$, and $b$ (see \autoref{fig:trans}).
  Its initial marking is expanded so that $\mathsf{sim}$ originally
  contains one token,  $\mathsf{sum}$ the sum $s_{m_0}=\sum_p m_0(p)$ of
  all the tokens in the initial marking of $\mathcal{N}$, and
  $\mathsf{test}$ no token.  It simulates $\mathcal{N}$ while a token
  resides in $\mathsf{sim}$, and updates $\mathsf{sum}$ so that it
  contains at all times the sum of the tokens in all the places of
  $\mathcal{N}$.
  Transfer arcs are not an issue since they do not change this overall
  sum of tokens.  Nondeterministically, $\mathcal{N}'$ fires $t$, which
  removes $m(p)$ in each place $p$ of $\mathcal{N}$, one token from
  $\mathsf{sim}$, $s_m=\sum_p m(p)$ tokens from $\mathsf{sum}$, and places
  one token in~$\mathsf{test}$.

  Now, a token can appear in $\mathsf{test}$ if and only if a marking
  $m'$ larger than $m$ can be reached in $\mathcal{N}'$.  Furthermore,
  the distance $\sum_p m'(p)-m(p)$ is in $\mathsf{sum}$, so that $m$
  is reachable in $\mathcal{N}$ if and only if a marking with one
  token in $\mathsf{test}$ and no token in $\mathsf{sum}$ is reachable
  in~$\mathcal{N}'$.

  The latter condition is tested by having $a$ remove
  one token from $\mathsf{test}$ and put one token in $\mathsf{sum}$ and
  one back in $\mathsf{test}$, and $b$ remove one from $\mathsf{sum}$
  and $\mathsf{test}$ and put them back.  Set
  $I\eqdef\{(a,b),(b,a)\}$; $\mathcal{N}'$ is $I$-diamond.  The transition
  sequence $ab$ can be fired if and only if there if a token in
  $\mathsf{test}$, but $ba$ further requires $\mathsf{sum}$ not to be
  empty.  Thus $a$ and $b$ do not commute if and only if $m$ is
  reachable in~$\mathcal{N}$.
\end{proof}

\paragraph{Foata Normal Form}
Let us assume an arbitrary linear ordering $<$ on $\Sigma$.  For an
independence relation $I$, we denote by $\clq(I)$ the set of cliques of
$I$, i.e.\
\begin{equation*}
  \clq(I)\eqdef\{C\subseteq\Sigma\mid\forall a,b\in C, (a,b)\in I\}\;.
\end{equation*}
We further introduce a homomorphism
$\nu:2^\Sigma\rightarrow\Sigma^\ast$ by
\begin{align*}
  \nu(\{a_1,a_2,\dots,a_k\})&= a_1a_2\cdots a_k&\!\text{if
  }a_1<a_2<\cdots<a_k.
\end{align*}

An infinite word $\sigma$ in $\Sigma^\omega$ is in \emph{Foata normal
form} \citep[see e.g.][]{inftr} if there is an infinite decomposition
$\sigma=\nu(C_0)\nu(C_1)\cdots$ with each $C_i$ in $\clq(I)$, and for
each $a$ in $C_i$, there exists $b$ in $C_{i-1}$ such that $(a,b)$ is in
$D$.  As indicated by its name, for any word $\sigma$ in
$\Sigma^\omega$, there exists a unique word $\fnf(\sigma)$ in Foata
normal form such that $\sigma\cIlim\fnf(\sigma)$.  For instance
$\fnf((aab)^\omega)=(ab)^\omega$ for $I=\{(a,b),(b,a)\}$.

Let us finally define the \emph{normalizing language} $N_I$ of $I$ as
the set of all infinite words in Foata normal form.  The following lemma
shows that $N_I$ is very well behaved, being recognized by a
deterministic B\"uchi automaton $\mathcal{B}_I$ with only accepting
states.  Thus its synchronous product with a WSTS $\mathcal{S}$ does not
require the addition of an acceptance condition:
$T_\omega(\mathcal{S}\times\mathcal{B}_I)=T_\omega(\mathcal{S})\cap
N_I$.
\begin{lemma}\label{lem:fnfsafety}
  Let $I$ be an independence relation on $\Sigma$.  Then $N_I$ is a
  topologically closed $\omega$-regular language.
\end{lemma}
\begin{proof}
  The topologically closed $\omega$-regular languages, aka ``safety''
  languages, are the languages recognized by finite deterministic
  B\"uchi automata with only accepting states.  We provide such an
  automaton $\mathcal{B}_I=\tup{Q,\Sigma,q_0,\delta,Q}$ such that
  $L(\mathcal{B}_I)=N_I$.

  Set $Q\eqdef\{q_0\}\cup(\clq(I)\cup\{\Sigma\})\times\clq(I)\times\Sigma$.
  We define $\delta(q_0,a)$ as $(\Sigma,\{a\},a)$ for all $a$ in $\Sigma$;
  for all $C_1$ in $\clq(I)\cup\{\Sigma\}$, $C_2$ in $\clq(I)$, $a$, $b$
  in $\Sigma$, we define $\delta((C_1,C_2,a),b)$ by
  \begin{equation*}
    \begin{cases}(C_1,C_2\cup\{b\},b)&\text{if
    }a<b,~\exists d\in C_1, (b,d)\in D,\\&\phantom{\text{if }}\text{and }\forall d\in C_2, (b,d)\in I,\\
    (C_2,\{b\},b)&\text{if }\exists d\in C_2,(b,d)\in D\;.\end{cases}
  \end{equation*}
  The automaton simultaneously checks that consecutive cliques enforce
  the Foata normal form, and that the individual letters of each clique
  are ordered according to~$<$.
\end{proof}

\subsection{The Alternating Bit Protocol}\label{ax:abp}
The \emph{Alternating Bit Protocol} (ABP) is one of the oldest case
studies~\citep{boch80}.  It remains interesting today because no
complete and automatic procedure exists for its verification.  It can
be nicely modeled as a lossy channel system \citep[see][and the next
discussion ``A Quick Tour'']{fwlcs}, but even in this representation,
liveness properties cannot be checked.  We believe it provides a good
illustration of the kind of issues that make a system trace unbounded,
which we categorize into commutativity issues, which we tackle through
normalization, and main control loop issues, which we avoid by
bounding the number of sessions.

\paragraph{A Quick Tour}
If the ABP is modeled as a fifo automaton (in fact two finite automata
communicating through two fifo queues), then all non-trivial properties
are undecidable, because fifo automata can simulate Turing
machines \citep[see e.g.][]{BZ83}.  Nevertheless, several classes of
fifo automata have been studied in the literature, often with decidable
reachability problems:
\begin{itemize}
\item One may observe that for any control state $q$ of this particular
fifo automaton, the language of the two fifo queues is
\emph{recognizable} (as a subset of $\{q\}\times A^\ast\times B^\ast$
where $A$ and $B$ are the alphabets of the queues).  \Citet{pachl} has
shown that reachability and safety are then decidable.  But this
recognizability property itself is in general undecidable.
\item One may also observe that the languages of the fifo queues contents
are \emph{bounded}~\citep{FC-Petri87}, and then one may simulate the
fifo automaton with a Petri net and decide reachability.  Again, this
subclass of fifo automata is not recursive.
\item Yet another way is to use \emph{loop acceleration} with
  QDDs~\citep{qdd} or more generally CQDDs~\citep{cqdd} as symbolic
  representations, and to observe that the reachability set is CQDD
  computable; but still without termination guarantee when applied to
  non-flat systems.
\end{itemize}
Neither of these techniques is fully automatic nor allows to check
liveness properties.

\begin{figure}[t!]
  \centering
  \begin{tikzpicture}[->,>=stealth',shorten >=1pt,initial text=,%
                    node distance=2cm,on grid,semithick,auto,
                    inner sep=2pt]
    \node[state,initial](0){0};
    \node[state,right=of 0](1){1};
    \node[state,below=of 1](2){2};
    \node[state,below=of 0](3){3};
    \node[below right=1cm and 1.3cm of 3]{Sender};
    \path[every initial by arrow/.style={color=gray}]
      node[state,color=gray,initial,right=2.5cm of 1](a){0};
    \node[state,color=gray,right=of a](b){1};
    \node[state,color=gray,below=of b](c){2};
    \node[state,color=gray,below=of a](d){3};
    \node[below left=1cm and 1.3cm of c,color=gray]{Receiver};
    \path[every node/.style={font=\footnotesize}]
      (0) edge node{$\mathsf{snd}$} (1)
      (1) edge[loop above] node{$c_M!0$} ()
      (1) edge node{$c_A?0$} (2)
      (2) edge node{$\mathsf{snd}$} (3)
      (3) edge[loop below] node{$c_M!1$} ()
      (3) edge node{$c_A?1$} (0);
    \path[every node/.style={font=\footnotesize,color=gray}]
      (a) edge[loop above,color=gray] node{$c_A!1$} ()
      (a) edge[color=gray] node{$c_M?0$} (b)
      (b) edge[color=gray] node{$\mathsf{rcv}$} (c)
      (c) edge[loop below,color=gray] node{$c_A!0$} ()
      (c) edge[color=gray] node{$c_M?1$} (d)
      (d) edge[color=gray] node{$\mathsf{rcv}$} (a);
  \end{tikzpicture}
  \caption{\label{fig:abp}The Alternating Bit Protocol.}
\end{figure}
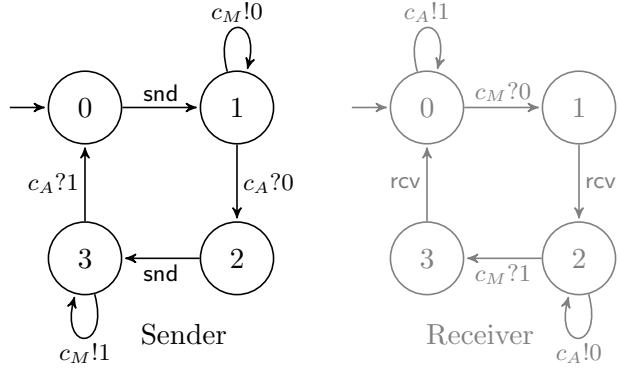
The most effective approach is arguably to model the ABP as a lossy
channel system (see \autoref{fig:abp}); reachability and safety are then
decidable, but liveness remains undecidable.  Furthermore, a forward
analysis using SREs as symbolic representations---as performed by a tool
like \textsc{TReX}---, will terminate and construct a finite symbolic graph (for
the verification of safety properties) \citep{fwlcs}: indeed, the ABP is
\emph{cover flattable}, but unfortunately this property is in general
undecidable.

\paragraph{Verification}
We model the ABP as two functional lossy channel systems (Sender and
Receiver) that run in parallel, and communicate through two shared
channels $c_M$ for messages and $c_A$ for acknowledgments.  Our
correctness property is whether each sent message (proposition
$\mathsf{snd}$) is eventually received (proposition~$\mathsf{rcv}$):
\begin{equation}\label{eq:liveabp}
  \mathsf{G}(\mathsf{snd}\Rightarrow\mathsf{X}(\neg\,\mathsf{snd}\mathrel{\mathsf{U}}\mathsf{rcv}))\;,\tag{$\varphi_\text{ABP}$}
\end{equation}
under a weak fairness assumption (every continuously firable transition
is eventually fired).

The full system is displayed for its useful accessible part in
\autoref{fig:abpsync}, with Receiver's transitions in grey.  This
system is clearly not trace bounded, thus we cannot apply
\autoref{thlive} alone.

\begin{figure}[t]
  \centering
  \hspace*{-.7ex}%
  \begin{tikzpicture}[->,>=stealth',shorten >=1pt,initial text=,%
                    node distance=1.7cm,on grid,semithick,auto,
                    inner sep=2pt]
    \node[state,initial](00){0{\color{gray}0}};
    \node[state,right=of 00](10){1{\color{gray}0}};
    \node[state,right=of 10](11){1{\color{gray}1}};
    \node[state,right=of 11](12){1{\color{gray}2}};
    \node[state,below=2cm of 12](22){2{\color{gray}2}};
    \node[state,left=of 22](32){3{\color{gray}2}};
    \node[state,left=of 32](33){3{\color{gray}3}};
    \node[state,left=of 33](30){3{\color{gray}0}};
    \path[every node/.style={font=\footnotesize}]
      (00) edge[loop above,color=gray] node[color=gray]{$c_A!1$} ()
      (00) edge node{$\mathsf{snd}$} (10)
      (10) edge[loop above,color=gray] node[color=gray]{$c_A!1$} ()
      (10) edge[loop below] node{$c_M!0$} ()
      (10) edge[color=gray] node[color=gray]{$c_M?0$} (11)
      (11) edge[loop above] node{$c_M!0$} ()
      (11) edge[color=gray] node[color=gray]{$\mathsf{rcv}$} (12)
      (12) edge[loop above] node{$c_M!0$} ()
      (12) edge[loop right,color=gray] node[color=gray]{$c_A!0$} ()
      (12) edge node{$c_A?0$} (22)
      (22) edge[loop right,color=gray] node[color=gray]{$c_A!0$} ()
      (22) edge node{$\mathsf{snd}$} (32)
      (32) edge[loop above,color=gray] node[color=gray]{$c_A!0$} ()
      (32) edge[loop below] node{$c_M!1$} ()
      (32) edge[color=gray] node[color=gray]{$c_M?1$} (33)
      (33) edge[loop below] node{$c_M!1$} ()
      (33) edge[color=gray] node[color=gray]{$\mathsf{rcv}$} (30)
      (30) edge[loop below] node{$c_M!1$} ()
      (30) edge[loop left,color=gray] node[color=gray]{$c_A!1$} ()
      (30) edge node{$c_A?1$} (00);
  \end{tikzpicture}
  \caption{\label{fig:abpsync}Synchronized view of the ABP.}
\end{figure}

\subsection{Trace Bounded Modulo $I$}\label{sub:norm}
The search for increasing forks on the ABP successively finds four
witnesses of trace unboundedness in states $10$, $12$, $32$, and $30$,
where at each occasion two competing elementary loops can be fired.
Thankfully, all these loops commute, because they involve two
different channels.  Our goal is to transform our system in order to
remove these forks, while maintaining the ability to verify
\eqref{eq:liveabp}.

\begin{definition}\label{defmod}%
  A WSTS $\mathcal{S}$ is \emph{trace bounded modulo $I$} an
  independence relation, if $T_\omega(\mathcal{S})$ is $I$-closed and
  the set of finite prefixes of the normalized language
  $T_\omega(\mathcal{S})\cap N_I$ is trace bounded.
\end{definition}
By \autoref{lem:fnfsafety}, we can construct a cd-WSTS $\mathcal{S}'$
for $T_\omega(\mathcal{S})\cap N_I$, and decide whether it is trace
bounded thanks to \autoref{cordec}.  Thus trace boundedness modulo $I$
is decidable for $I$-closed WSTS.

Finally, provided the language $L(\neg\varphi)$ of the property to
verify is also $I$-closed, the normalized system and the original
system are equivalent when it comes to verifying $\varphi$.  Indeed,
we can generalize \autoref{thlive} to trace bounded modulo $I$
cd-WSTS and $I$-closed $\omega$-regular languages:
\begin{theorem}\label{propmod}%
  Let $I$ be an independence relation, $\mathcal{S}$ be a trace bounded
  modulo $I$ cd-WSTS, and $L$ an $I$-closed $\omega$-regular language,
  all three on $\Sigma$.  Then it is decidable whether
  $T_\omega(\mathcal{S})\cap L$ is empty.
\end{theorem}
\begin{proof}
  By \autoref{lem:fnfsafety}, we can construct a cd-WSTS
  $\mathcal{S}'$ for $T_\omega(\mathcal{S})\cap N_I$, which will be
  trace bounded by hypothesis.  Wlog., we can assume that we have a DFA
  with a Rabin acceptance condition for $L$, and can apply
  \autoref{thlive} to decide whether $T_\omega(\mathcal{S}')\cap
  L=\emptyset$.
  
  It remains to prove that
  \begin{equation*}
  T_\omega(\mathcal{S})\cap L=\emptyset\text{  iff 
  }T_\omega(\mathcal{S}')\cap L=\emptyset\;.
  \end{equation*}
  Obviously, if $T_\omega(\mathcal{S})\cap L$ is empty, then the same
  holds for $T_\omega(\mathcal{S}')\cap L$.  For the converse, let
  $\sigma$ be a word in $T_\omega(\mathcal{S})\cap L$.  Then, since
  $\mathcal{S}$ is $I$-closed, $\fnf(\sigma)$ also belongs to
  $T_\omega(\mathcal{S})$ and to $N_I$, and thus to
  $T_\omega(\mathcal{S}')$.  And because $L$ is $I$-closed,
  $\fnf(\sigma)$ further belongs to $L$, hence to
  $T_\omega(\mathcal{S}')\cap L$.
\end{proof}

Once our system is normalized against partial commutations, the only
remaining source of trace unboundedness is the main control loop.  By
bounding the number of sessions of the protocol, i.e.\ by unfolding
this main control loop a bounded number of times, we obtain a trace
bounded system.

This transformation would disrupt the verification of \eqref{eq:liveabp},
if it were not for the two following observations:
\begin{enumerate}
\item The full set of all reachable configurations is already explored
  after two traversals of the main control loop.  This is established
  automatically thanks to \autoref{cor:cover} on the 2-unfolding
  of the normalized ABP, which is a trace bounded cd-WSTS.  Thus any
  possible session, with any possible reachable initial configuration,
  can already be exhibited at the second traversal of the system.
\item Our property \eqref{eq:liveabp} is intra-session: it only requires
  to be tested against any possible session.
\end{enumerate}
The overall approach, thanks to the concept of trace boundedness modulo
partial commutations, thus succeeds in reducing the ABP to a trace bounded
system where our liveness property can be verified.

\section{Trace Boundedness is not a Weakness}\label{sec:appl}
\begin{table}[t]%
  \caption{\label{tab:decsum}Some decidability results
    for selected classes of cd-WSTS---Petri nets (PN), affine counter
    systems (ACS), and functional lossy channel systems (LCS)---in the
    general and trace bounded cases (t.b.).}%
  \centering{\small
  \begin{tabular}{lcccccc}
    \toprule
    & PN & t.b.\ PN& ACS & t.b.\ ACS & LCS & t.b.\ LCS\\
    \midrule
    Reachability& \tickYes & \tickYes & \tickNo & \tickNo & \tickYes &
    \tickYes\\
    $\mathsf{Post}^\ast$ inclusion%
    & \tickNo & \tickYes & \tickNo &
    \tickNo & \tickNo & \tickYes\\
    Liveness & \tickYes & \tickYes & \tickNo & \tickYes & \tickNo &
    \tickYes\\
    \bottomrule
  \end{tabular}}
\end{table}
To paraphrase the title \emph{Flatness is not a
  Weakness}~\citep{flatltl}, trace boundedness is a powerful property for the
analysis of systems, as demonstrated with the termination of forward
analyses and the decidability of $\omega$-regular
properties for trace bounded WSTS (see also \autoref{tab:decsum})---and is implied by
flatness.  More examples of its interest can be found in the recent
literature on the verification of multithreaded programs, where
trace boundedness of the context-free synchronization languages yields
decidable reachability~\citep{vineet,parikhb}. 

Most prominently, trace boundedness has the
considerable virtue of being decidable for a large class of systems,
the $\infty$-effective complete deterministic WSTS.  There is
furthermore a range of unexplored possibilities beyond partial
commutations (starting with semi-commutations or contextual
commutations) that could help turn a system into a trace bounded one.

\subsection*{Acknowledgments}
We thank the anonymous reviewers for their careful reading, which
improved the paper.

\bibliographystyle{abbrvnat}
\bibliography{journalsabbr,bounded}
\end{document}